\documentclass[12pt]{article}\usepackage[]{graphicx}\usepackage[]{color}
\makeatletter
\def\maxwidth{ %
  \ifdim\Gin@nat@width>\linewidth
    \linewidth
  \else
    \Gin@nat@width
  \fi
}
\makeatother

\definecolor{fgcolor}{rgb}{0.345, 0.345, 0.345}

\usepackage{framed}
\makeatletter
 {\par\unskip\endMakeFramed%
 \at@end@of@kframe}
\makeatother

\definecolor{shadecolor}{rgb}{.97, .97, .97}
\definecolor{messagecolor}{rgb}{0, 0, 0}
\definecolor{warningcolor}{rgb}{1, 0, 1}
\definecolor{errorcolor}{rgb}{1, 0, 0}
\newenvironment{knitrout}{}{} 

\usepackage{alltt}

\usepackage{amssymb,amsthm,natbib,fullpage,amsmath,hyperref}

\newcommand{\Exp}[1]{{\mathbb E}[ #1 ]}
\newcommand{\Var}[1]{{\mathbb V}[ #1 ]}

\newcommand{\bl}[1]{{\mathbf #1}}
\newcommand{\bs}[1]{\boldsymbol #1}

\newtheorem{thm}{Theorem}
\newtheorem{cor}{Corollary}

\title{Smaller $p$-values via indirect information}
\author{Peter Hoff}
\date{\today}
\IfFileExists{upquote.sty}{\usepackage{upquote}}{}
\begin{document}
\maketitle

\begin{abstract} 
This article develops  $p$-values 
for evaluating means of normal populations
that make use of indirect or prior
information.  
A $p$-value of this type is based on a biased test statistic that is 
optimal on average with respect to a probability distribution 
that encodes indirect information about the 
mean parameter, resulting in a smaller $p$-value if the indirect information 
is accurate. 
In a variety of multiparameter settings, we show how to adaptively 
estimate the indirect information  for each mean parameter while still maintaining uniformity of the $p$-values under their 
null hypotheses. This is done using a linking model through which 
indirect information about the mean of one population may be obtained 
from the data of other populations. Importantly, the 
linking model does not need to be correct to maintain the uniformity of the 
$p$-values under their null hypotheses. 
This methodology is illustrated in several data analysis scenarios, 
including small area inference, spatially arranged populations, 
interactions in linear regression, 
and generalized linear models. 

\end{abstract}

\section{Introduction}
Many statistical procedures are built upon the 
evaluation of evidence that  one or more unknown scalar quantities 
are not equal to some null value, say zero. 
A controversial yet ubiquitous measure of evidence 
that a scalar $\theta$ is not zero 
is the 
$p$-value, a function of the data that is uniformly distributed on $(0,1)$ if 
$\theta=0$. 
One way to construct a $p$-value is with a collection of tests that, 
for each value of $\alpha\in (0,1)$, includes a single test 
with type I error rate equal to $\alpha$. Given such a collection, 
the $p$-value can be defined as the smallest value of $\alpha$ for which 
the corresponding test rejects $H$, assuming the tests satisfy a
monotonicity property \citep{dickhaus_2014}. 
As such, a $p$-value will likely be small if the true value of $\theta$ 
is one for which the corresponding 
tests have a high probability of rejecting $H$. 

If indirect or prior information about 
$\theta$ is available, then 
tests may be constructed that have high power at values of 
$\theta$ that are 
likely, at the expense of having lower power elsewhere. 
If the indirect information is accurate, then a $p$-value derived 
from such tests will generally be smaller 
than a standard $p$-value that does not make use of the information. 
In this article, $p$-values are developed that 
correspond to optimal tests when the indirect information 
is summarized with a probability distribution over $\theta$. 
Specifically, it is shown that if the direct information about $\theta$ is from  data $Y\sim N(\theta,\sigma^2)$, and the indirect 
information about $\theta$ is encoded 
with a normal distribution $\pi$, then 
the test statistic with optimal 
power on average with respect to $\pi$  may be derived analytically, and the
corresponding $p$-value may be expressed very simply as 
$1-|\Phi(Y/\sigma + 2\mu\sigma/\tau^2) - \Phi(-Y/\sigma)|$, 
where $\mu$ and $\tau^2$ are the mean and variance of $\pi$, and 
$\Phi$ is the standard normal  cumulative distribution function.
A similar $p$-value for the case that $\sigma^2$ is estimated from data is also 
derived. These $p$-values are uniformly distributed if $\theta=0$, regardless 
of the values of  $\mu$ and $\tau^2$. 

We refer to such a $p$-value as being  ``frequentist, assisted by Bayes'', or FAB. It is frequentist in the sense that it has guaranteed sampling 
properties (uniformity under the null distribution), and it is Bayesian 
in the sense that the corresponding tests maximize expected power, on average across $\theta$-values with respect to a probability distribution. 
The basic idea of using a Bayesian criteria to select a frequentist procedure
goes back at least to 
\citet{pratt_1963}, who constructed a confidence interval for a normal 
mean that has minimum expected width with respect to a distribution 
over the possible values of the mean. 
In fact, the FAB $p$-value
 may alternatively be derived by finding the smallest $\alpha$ for which 
Pratt's $1-\alpha$ confidence interval does not contain zero. 
Related to this, 
\citet{good_crook_1974} compared Bayes factors to a variety of other statistics
for evaluating equiprobability of multinomial probabilities, and 
referred to the use of such statistics and consideration of average 
power as a ``Bayes/non-Bayes compromise'' (see also \citet{good_1992} for 
other such compromises). 
More recently, \citet{servin_stephens_2007} used a Bayes factor to 
evaluate a global test of association between a phenotypic outcome 
and several genetic markers. A 
null distribution for their global test is obtained via permutation, 
as in this case the null hypothesis  
corresponds to an exchangeability assumption. 
Going the other direction, 
\citet{wakefield_2009} and \citet{benjamin_berger_2019}
study how standard $p$-value criteria should be
modified in order to produce inferences that resemble those that would be 
obtained using Bayes factors.

In addition to their use for measuring evidence, $p$-values are
also used as inputs into statistical decision procedures
that control various frequentist error rates. For example,
rejecting the hypothesis that $\theta=0$ whenever the $p$-value is
below some threshold $\alpha$ is a procedure that
of course maintains a
false rejection rate of $\alpha$. In multiparameter analyses,
$p$-values are often combined in ways to evaluate a global
hypothesis  \citep{birnbaum_1954,heard_rubin-delanchy_2018},
to adaptively estimate a threshold that maintains control of 
the false discovery rate
 \citep{benjamini_hochberg_1995}, or to maintain such a false discovery rate
among multiple groups of hypotheses \citep{barber_ramdas_2017}.
Frequentist multiparameter inference procedures such as these 
take as their input a list of 
$p$-values, typically without specifying how the $p$-values are constructed. 
In contrast, the focus of this article is on how, 
in multiparameter scenarios, 
adaptive FAB $p$-values that are smaller on average than standard $p$-values
may be constructed 
by sharing information across 
an entire dataset. 
Specifically, 
while the FAB $p$-value 
can be as small as half the usual $p$-value derived from a 
uniformly most powerful unbiased (UMPU) test, to realize this 
gain the indirect information needs to be in accord with the actual 
value of $\theta$. Scenarios where such accordance can be obtained
statistically include multiparameter problems 
where the parameters 
are believed to be similar to one another in some way. 
For example, if the data consist of independent samples 
from  each of several groups, 
then the indirect information for the parameter of one 
group may be inferred from the data of the others 
using a model for across-group heterogeneity of the parameters, which we 
refer to as a linking model. If the linking model is precise then the FAB 
$p$-values will be smaller than the UMPU $p$-values, on average across groups. If the linking model is diffuse, then the FAB $p$-values will be similar to the UMPU $p$-values, as the latter are a special case of the former. 
Importantly, the resulting FAB $p$-values are adaptive in that 
the parameters of the linking model are estimated from the data, 
and robust in the sense that the uniformity of the $p$-values 
under their null hypotheses 
does not depend on the linking model being correct.

The utility of indirect information for multiparameter 
inference problems has been well-recognized 
\citep{ghosh_rao_1994,efron_2010}. While
hierarchical models and empirical Bayes procedures 
provide a powerful and flexible class of methods for information-sharing in 
such problems, 
their 
performance guarantees typically hold globally, on average across parameter 
values rather than for each parameter individually. 
For parameter-specific inference, 
the basic strategy of using indirect information and a linking model 
to choose statistical procedures with parameter-specific 
frequentist guarantees was used previously in
 \citet{yu_hoff_2018},
\citet{hoff_yu_2019} and
 \citet{burris_hoff_2018} 
to obtain adaptive versions of Pratt's confidence interval 
for multiparameter settings. In the present article, the use of adaptive tests 
to construct a FAB $p$-value for each parameter is similar to the adaptive empirical likelihood ratio 
used by \citet{storey_2007}, which combines
 data from multiple experiments 
to form an empirically estimated significance threshold, 
and is then individually applied to each experiment. 
Storey's empirical likelihood ratio 
corresponds in some sense to a 
semiparametric exchangeable linking model, while the parametric 
approach proposed here permits other types of information to be shared, such as spatial 
information or parameter-specific explanatory variables. 
Other recently-developed 
evidence measures that incorporate auxiliary information include
the ``$s$-value'' \citep{grunwald_deheide_koolen_2019}
which can be viewed as being derived from a Bayes factor but has frequentist 
guarantees for 
sequential testing; 
and  the 
``skeptical $p$-value'' \citep{held_2018}, which combines information from 
an original study with that from a replication study.

The FAB $p$-value is derived and studied in Section 2. 
Section 3 illustrates how adaptive FAB $p$-values 
may be constructed in two data analysis scenarios
in which there is independent direct data for each parameter. 
These include 
small area inference using the Fay-Herriot model
\citep{fay_herriot_1979}, and inference for 
a sequence of means using data from a hidden Markov process. 
The methodology is extended in Section 4 to situations 
where the direct data for each parameter may be dependent, such as 
estimates of linear regression  coefficients. 

This article primarily focuses on scenarios involving normally (or $t$) distributed 
test statistics and normal prior distributions, as in these cases the 
FAB tests and $p$-values have surprisingly simple forms. To accommodate more general 
data-analysis scenarios, approximate FAB $p$-values
for asymptotically normal estimators are developed in Section 4, where it is 
proven that these $p$-values are 
asymptotically uniform under the null distribution. These approximate FAB $p$-values 
can be used, for example, to evaluate coefficients in 
widely-used 
generalized linear models 
such as logistic or Poisson regression. 
A discussion follows in Section 5. 

\section{FAB $p$-values}  

\subsection{Bayes-optimal tests and $p$-values}
Consider performing a level-$\alpha$ test of  $H:\theta=0$ 
based on the observation of 
$Y\sim N(\theta, \sigma^2)$, with $\sigma^2$ known.  
Suppose additionally that indirect information about 
$\theta$ is available, via prior knowledge or data independent 
of $Y$, that is encoded with a distribution having density $\pi$.  
The average power of a test function $f:\mathbb R \rightarrow [0,1]$ 
with respect to $\pi$ is given by 
\begin{align}
 \int \Exp{ f(Y) | \theta }  \pi(\theta) \, d\theta&= 
  \int \int f(y)  p(y|\theta)  \pi(\theta) \, dy \, d\theta    \nonumber \\
 &= \int f(y) \left (  \int  p(y|\theta)  \pi(\theta) \, d\theta \right ) dy   \nonumber  \\
 &= \int f(y) \, p_\pi(y) \, dy,  \label{eqn:priorpower}
\end{align}
where $p(y|\theta)$ is the $N(\theta,\sigma^2)$ density and 
$p_\pi$ is the marginal density of $Y$ under $\pi$.

According to our indirect information, 
the level-$\alpha$ test that has the highest probability of
rejecting $H$ is 
the function $f_\pi$ that maximizes (\ref{eqn:priorpower}) among all 
functions $f:\mathbb R \rightarrow [0,1]$ that satisfy 
 $\Exp{f(Y)|\theta} = \alpha$ when $\theta=0$. 
The corresponding test has a frequentist type I error rate of $\alpha$,
but it is also the Bayes-optimal level-$\alpha$ test if 
$\pi$ is viewed as a prior density. 
As such, we describe
this test as being ``frequentist, assisted by Bayes'', or FAB. 

By the Neyman-Pearson lemma, the FAB test is to reject $H$ if 
 $p_\pi(Y)/p_0(Y)$ exceeds its upper $\alpha$ quantile
under $H$, where $p_0(y)$ is the $N(0,\sigma^2)$ null density.  
If $\pi$ is a normal density, then the 
test has a simple analytic form. 
\begin{thm}
\label{thm:fabtest}
The level-$\alpha$ FAB test of $H$ corresponding to a $N(\mu,\tau^2)$ distribution for $\theta$ 
rejects when $|Y + \mu\sigma^2/\tau^2|>c$,  where 
  $c$ is the solution to 
$[ \Phi(c/\sigma+\mu \sigma/\tau^2) +
                \Phi(c/\sigma-\mu \sigma/\tau^2) ]/2  = 1-\alpha/2$   
 and $\Phi$ is the standard normal cumulative distribution function. 
\end{thm}
\begin{proof} 
If $Y|\theta\sim N(\theta,\sigma^2)$ and
$\theta\sim N(\mu,\tau^2)$ then the marginal distribution 
of $Y$ is $N(\mu, \sigma^2+\tau^2)$. 
The likelihood ratio is then 
\[ \frac{p_\pi(y)}{p_0(y) } = \exp( - [ (y^2- 2\mu y +\mu^2)/(\sigma^2+\tau^2) - y^2/\sigma^2 ]/2  ).\]
Some algebra shows that this is monotonically increasing in 
$|y+ \mu \sigma^2/\tau^2 |$. 
\end{proof}

Different prior distributions for $\theta$ typically lead to different 
Bayes-optimal test statistics. 
However, 
one particular prior distribution, used in an example in Section 4, leads 
to the same test statistic as a normal prior distribution. 

\begin{cor}
The level-$\alpha$ FAB test of $H$ when the 
prior distribution for $\theta$ is a
mixture of a 
$N(\mu,\tau^2)$ distribution with a point-mass at zero 
is the same as the test when the prior distribution is 
$N(\mu,\tau^2)$.  
\label{cor:mix}
\end{cor}
\begin{proof}
Let $p_w(y)$ be the marginal density of $Y$ under the 
mixture prior distribution with weight $w$ on the normal component. 
Then $p_w(y) = w p_\pi(y) + (1-w)p_0(y)$, where $p_\pi$ and $p_0$ 
are as in the proof of Theorem \ref{thm:fabtest}. The likelihood ratio is 
$[ w p_\pi(y) + (1-w) p_0(y) ] /p_0(y) $, which simplifies to 
 $w p_{\pi}(y) /p_0(y) + 1-w$. This is monotonically increasing in $p_\pi(y)/p_0(y)$, the 
likelihood ratio when $w=1$, and so the tests are the same. 
\end{proof}

The FAB $p$-value function is the function that maps 
potential values $y$ of $Y$ to the 
smallest value of $\alpha$ such that 
the FAB level-$\alpha$ test rejects when $Y=y$. 
It is also the probability that $|Y + \mu \sigma^2/\tau^2| > |y+ \mu \sigma^2/\tau^2|$  when   $Y \sim N(0,\sigma^2)$ 
(see, for example, \citet[Chapter~2]{dickhaus_2014}). 
The functional form of the FAB $p$-value is quite simple, and is easily 
derived using this latter 
characterization. 

\begin{thm}
The FAB $p$-value function may be written as 
$1- | \Phi([y+2 a]/\sigma  ) -\Phi(-y/\sigma) |$ 
where $a= \mu\sigma^2/\tau^2$, or as 
$p(z,b)=1- | \Phi(z+ b ) -\Phi(-z) |$, 
where $z=y/\sigma$ and $b=2 \mu\sigma/\tau^2$.
\label{thm:fabpvalue} 
\end{thm}
\begin{proof}
The $p$-value is 
\begin{align*} 
\Pr(| Y + \mu \sigma^2/\tau^2| > |y+ \mu \sigma^2/\tau^2| )&= 
\Pr( | Z+b/2| > |z+b/2| )  \\
&=  \Pr(  Z+b/2 < -|z+b/2| )  + \Pr( Z+b/2 > |z+b/2| ) \\
&= \Phi( -|z+b/2|-b/2) + 1- \Phi( |z+b/2|-b/2 ).
\end{align*}
First suppose $z+b/2>0$. In this case, the $p$-value is 
\[\Phi( -z-b) + 1- \Phi( z ) = 1-\Phi(z+b) + \Phi(-z) . 
\]
Note that $z+b/2>0$ implies 
$z+b> -z$, 
and so the $p$-value can be 
written 
$ 1- | \Phi(z+b) - \Phi(-z)|$. 
Now suppose $z+b/2<0$, so that the $p$-value is 
\[\Phi( z) + 1- \Phi( -z-b ) = 1-\Phi(-z) + \Phi(z+b) . 
\]
In this case, $z+b< -z$ so the $p$-value may be written 
as $1-| \Phi(-z)  - \Phi(z+b) |$, 
which is the same as in the case $z+b/2>0$.  So in either case, 
the $p$-value is 
$1- | \Phi(z+b ) -\Phi(-z) |$.
\end{proof}

\begin{figure}
\begin{knitrout}\footnotesize
\definecolor{shadecolor}{rgb}{0.969, 0.969, 0.969}\color{fgcolor}

{\centering \includegraphics[width=5.5in]{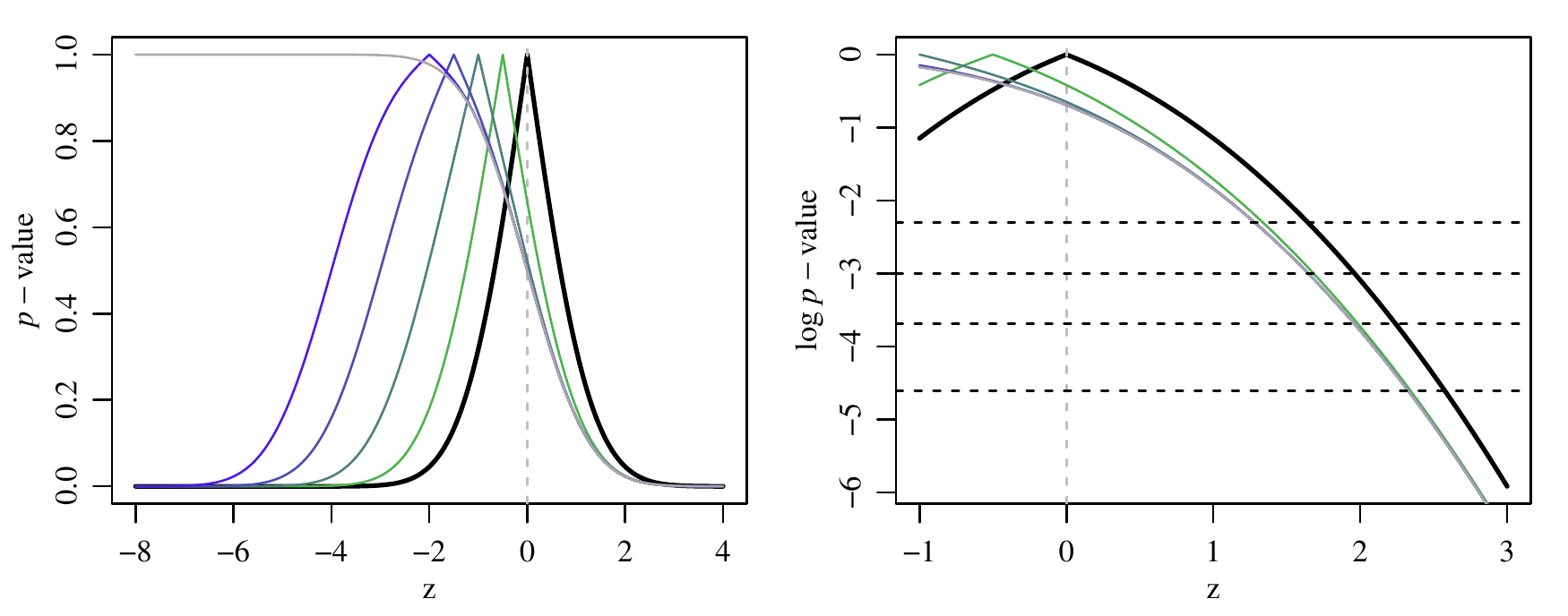} 

}

\end{knitrout}
\caption{FAB and UMPU $p$-value functions, for values of $b$ ranging from 
0 to 4 from right to left. The one-sided $p$-value function $1-\Phi(z)$ is given in gray. The figure on the right includes horizontal dashed lines 
at $\log .10$, $\log .05$, $\log .025$ and $\log .01$. }  
\label{fig:pFABz}
\end{figure}

Note that the standard $p$-value based on the UMPU test statistic 
$|Z|$ 
can be written as $p(Z,0)=1-|\Phi(Z)-\Phi(-Z)|$, and so 
is a special case of the FAB $p$-value with  $\mu=0$. 
For non-zero $\mu$, 
if the indirect information is accurate in the sense that $\theta$ and $\mu$ are of the same sign, 
then $Y$ and $2a$ (or $Z$ and $b$) will likely have the same sign, 
 making it probable that the absolute difference $|\Phi(Z+b) -  \Phi(-Z)|$
is larger than $|\Phi(Z) -  \Phi(-Z)|$, thereby  making the
FAB $p$-value smaller than the UMPU $p$-value. However, if 
the sign of $\mu$ does not match that of $\theta$, then this absolute 
difference 
will likely be smaller than  $|\Phi(Z) -  \Phi(-Z)|$,  making 
the FAB $p$-value larger than the UMPU $p$-value. This can be made more 
precise as follows:
\begin{thm} 
\label{prop:pbetter}
Let $Z\sim N(\theta,1)$. Then 
\[
 \Pr( p(Z,b) < p(Z,0) |\theta)> \Phi( \text{sign}(b) \times \theta ). 
\] 
\end{thm}
\begin{proof}
If $b$ is positive then $Z>0$ implies $p(Z,b)<p(Z,0)$, 
so $\Pr( p(Z,b) < p(Z,0)|\theta)> \Pr(Z>0|\theta) =  \Phi(\theta)$. 
Similarly, if $b$ is negative then 
 $\Pr( p(Z,b) < p(Z,0)|\theta)> \Phi(-\theta)$. 
Combining these gives the result. 
\end{proof}

Figure \ref{fig:pFABz}
plots the UMPU and FAB $p$-value functions 
for a variety of values of $b$. 
The FAB $p$-value functions are symmetric 
around $-b/2=  -\mu \sigma/\tau^2$, and so if 
$\mu$ is positive these $p$-values will be lower than the UMPU $p$-value 
if $z$ is 
positive, an can even be lower for a range of negative $z$-values. 
This is to be expected - the FAB $p$-value will be lower if 
$z$ and $\mu$ match, but can be higher otherwise. 
Like the UMPU $p$-value, 
for each finite value of $b$ the FAB $p$-value decreases to zero 
as $|z|$ increases. 
As shown in the right-side panel of Figure \ref{fig:pFABz}, 
for large $|z|$ 
the FAB $p$-value can be as small as half of the UMPU $p$-value, 
if  $z \mu>0$. For example, for the non-zero values of $b$ shown in the figure, 
when the UMPU $p$-value is .10, the 
FAB $p$-values are close to .05. 
However, the ratio of the FAB $p$-value to the UMPU $p$-value can 
be unboundedly large if there is a mismatch between the data and 
the indirect information
(that is, if $z \mu<0$). These comments are summarized as follows:

\begin{thm}
Let $p(z,b) = 1- | \Phi(z+b ) -\Phi(-z) |$. Then 
$p(z,0)$ is the UMPU $p$-value function, and for 
$b>0$, 
\begin{enumerate} 
\item $\lim_{z\rightarrow \infty} p(z,b)/p(z,0) = 1/2$;
\item $\lim_{z\rightarrow -\infty} p(z,b)/p(z,0) = \infty$.
\end{enumerate}
\end{thm}

\begin{proof}
For positive $z$ and $b$ the ratio of the $p$-values is 
\begin{align*}
  p(z,b)/p(z,0) &= \frac{ 1- \Phi(z+b) + \Phi(-z) }{ 1- \Phi(z) + \Phi(-z) } \\ 
\ &= \frac{ 1- \Phi(z+b) + \Phi(-z) }{ 2 \Phi(-z) }  \\
&= \frac{1}{2} \left(  1+ \frac{ 1 - \Phi(z+b) }{1-\Phi(z) } \right ), 
\end{align*}
which converges to 1/2 as $z\rightarrow \infty$ by L'H\^opital's rule. 
For  $z< -b/2 < 0$ the  ratio is 
\begin{align*}
\frac{1+\Phi(z+b) - \Phi(-z) }{ 1+\Phi(z)-\Phi(-z) }  
&= \frac{1}{2} \frac{\Phi(z+b) + \Phi(z)}{\Phi(z)} \\
& =  \frac{1}{2}\left ( 1+ \frac{\Phi(z+b)}{\Phi(z) } \right ), 
\end{align*}
which diverges to infinity as $z\rightarrow -\infty$. 
\end{proof}

The FAB $p$-value is based on a test that has better performance on 
one side of the real line than the other, suggesting a relationship to 
a one-sided test. Indeed, one-sided $p$-values are limiting cases of FAB $p$-values. 
\begin{thm} \ 
\begin{enumerate}
\item $\lim_{b\rightarrow \infty} p(z,b) = 1- \Phi(z) ; $
\item $\lim_{b\rightarrow -\infty} p(z,b) = \Phi(z). $ 
\end{enumerate}
\end{thm}
\noindent In other words, for large positive $b$ the FAB $p$-value approximates the $p$-value from the UMP tests 
of $H:\theta<0$ versus $K:\theta>0$, and vice versa for large 
negative $b$.

\subsection{Null and alternative distributions of FAB $p$-values}

Like UMPU $p$-values, FAB $p$-values are uniformly distributed 
under the null hypothesis. 
\begin{thm}
\label{thm:uniformity}
Let $Z\sim N(0,1)$. Then for any $b\in \mathbb R$,
$\Pr( 1- |\Phi( Z + b) - \Phi(-Z) | \leq  u ) =u$.
\end{thm}
\begin{proof} 
This is
guaranteed by the fact that the $p$-value is the probability of
obtaining a value of $|Z+ b/2|$
 as or more extreme than the one observed,
and that the null distribution of the test statistic is continuous
\citep[Chapter~2]{dickhaus_2014}. 
Alternatively, the result can be shown directly as follows:
Since the FAB $p$-value function $p(z,b)$ is symmetrically decreasing 
from $z=-b/2$, the probability that 
$p(Z,b)$ is less than $u$ is the probability that $Z$ lies outside 
the interval 
$(-b/2-c,-b/2+c)$, where $c>0$ satisfies $p(-b/2+c,b)=u$.  
Plugging this into the $p$-value function gives 
\begin{align*}
u &= 1- | \Phi(b/2 + c) - \Phi(b/2 -c )| \\
  &= 1-  \Phi(b/2 + c) + \Phi(b/2 -c ) \\
  & = \Phi(-b/2-c) + [1-\Phi(-b/2 +c)], 
\end{align*}
so the probability that $Z$ lies outside of $(-b/2-c,-b/2+c)$, 
and that $p(Z,b)<u$, is $u$. 
\end{proof}

\begin{figure}[ht]
\begin{knitrout}\footnotesize
\definecolor{shadecolor}{rgb}{0.969, 0.969, 0.969}\color{fgcolor}

{\centering \includegraphics[width=5.5in]{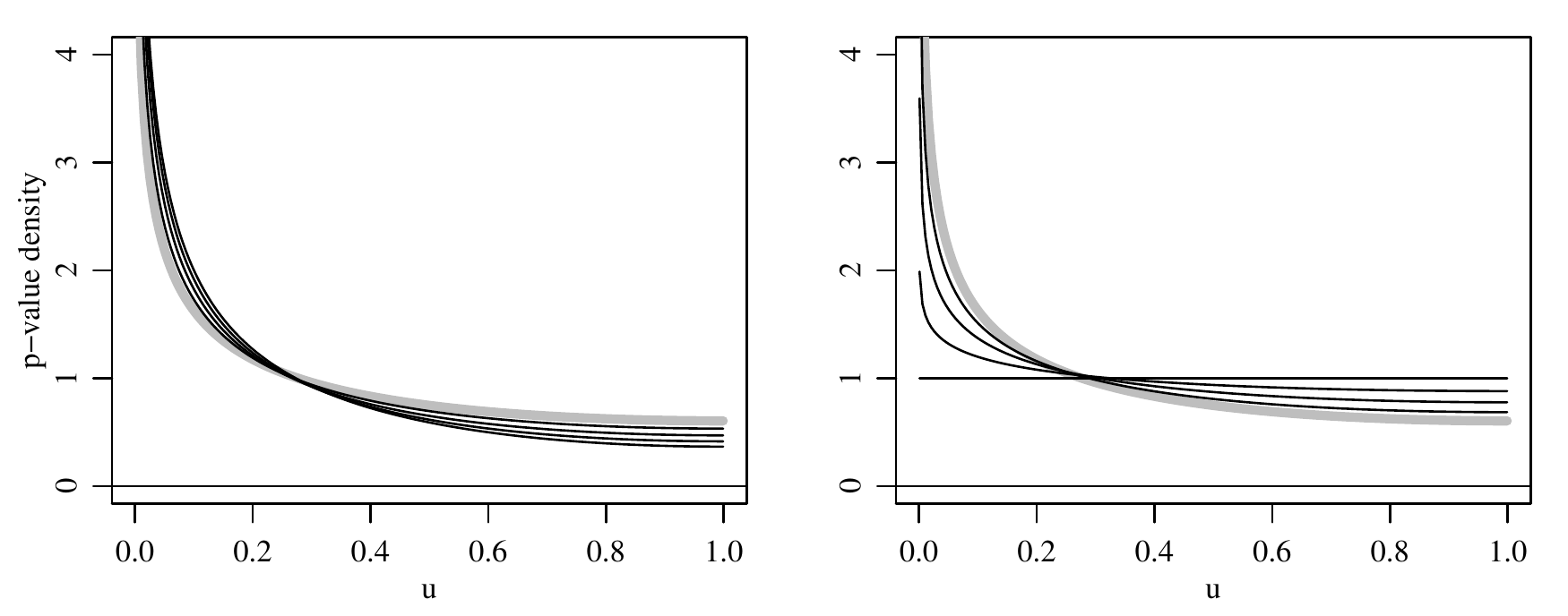} 

}

\end{knitrout}
\caption{Densities of $p$-values for  
 $b\in \{ 0,.25,.5,.75,1\}$, and 
$\theta=1$ (left)
 and $\theta=-1$ (right). 
The density 
corresponding to $b=0$ (the UMPU $p$-value) is displayed with a thick gray 
line.}  
\label{fig:densities}
\end{figure}

The CDF and density of the FAB $p$-value under non-zero values
of $\theta$ may be of use for power calculations and theoretical study.
To find these quantities,  recall that 
the FAB $p$-value function is symmetric around $-b/2$, and so 
each possible value $u$ of the $p$-value other than 1 is achieved by two
values of $z$, say $z_l$ and $z_h$, which satisfy $z_l< -b/2 < z_h$. 
These values 
 are the solutions to the equations
\begin{align}
  \Phi(z_h+b) - \Phi(-z_h) &  = 1-u   \label{eqn:zlu}\\  
   \Phi(-z_l)  -  \Phi(z_l+b) &  = 1-u. \label{eqn:zhu}
\end{align}
Since the $p$-value is monotonically increasing on $z< -b/2$
and decreasing on $z>-b/2$, the $p$-value will be less than $u$ if 
$z< z_l$ or $z> z_h$.
Now let $U = p(Z,b)$ where  $Z\sim N(\theta,1)$. 
Then the probability that $U\leq u$ is the probability
that $Z<z_l$ or $Z>z_h$, which is
\[
F_U(u) = \Pr( U\leq u) = \Phi( z_l-\theta ) + 1- \Phi(z_h-\theta). 
\]
Note that $z_l$ and $z_h$ are functions of
$u$, and  can be computed with a zero-finding algorithm.
Taking the derivative of $F_U$ gives the density of $U$, 
\[
  f_U(u) = \frac{ d F_U(u)}{du} = \phi(z_l-\theta)  \frac{d z_l}{du}   -
    \phi(z_h-\theta ) \frac{d z_h}{du}, 
\]
where $\phi$ is the standard normal density. 
The derivatives on the right side of the equation  can be found by implicit differentiation of
(\ref{eqn:zlu}) and (\ref{eqn:zhu}), giving
\[
f_U(u) = \frac{\phi(z_l-\theta)}{\phi(-z_l)+\phi(z_l+b)} +
         \frac{\phi(z_h-\theta)}{\phi(-z_h)+\phi(z_h+b)} . 
\]
Functions for computing the CDF and density of $U=p(Z,b)$ are available 
in the supplementary material for this article.

Figure  \ref{fig:densities} displays the 
density of 
$p(Z,b)$ for various values of $b$ and under various 
alternative values of $\theta$. 
If $b$ and 
$\theta$ are 
both positive (or are both negative) then the FAB $p$-value density 
is more concentrated  on small values than the density of the UMPU $p$-value. 
However, if there is a mismatch between $\theta$ and $b$ then the FAB $p$-value is expected to be larger. 

\subsection{FAB $p$-values from $t$-statistics} 
In most applications the value of $\sigma^2$ is 
unknown and must be estimated from the data. 
For example, suppose $Y_1,\ldots, Y_n \sim $ i.i.d.\ $N(\theta,\sigma^2)$ with $\sigma^2$ unknown and let $S$ be the sample standard deviation. 
If $\theta=0$ and $n$ is large then $T=\sqrt{n}\bar Y/S$ will be approximately standard normal and $1-|\Phi(T+b) - \Phi(T)|$ will be approximately uniform. 
If the sample size is insufficient to 
justify this approximation, then a $t_{n-1}$ null distribution should be used 
for construction of the $p$-value. 
Like the standard normal null distribution, 
$t$ distributions are continuous and symmetric about zero, 
and a $p$-value function may be defined as in the normal case.  
In fact, such a function may be similarly defined for any continuous null 
distribution that is symmetric about zero. 

\begin{thm}
\label{thm:gfabp}
Let $T\sim F$, where $F$ is the CDF
of any continuous distribution that is symmetric about zero.
Then for any $b\in \mathbb R$, 
$\Pr(  1-|F(T+ b) - F(-T)| < u ) = u $. 
\end{thm}
\noindent The proof is the same as that of Theorems \ref{thm:fabpvalue} and 
  \ref{thm:uniformity} but with 
$\Phi$ replaced by $F$.  

Suppose that we have a normal estimator 
$\hat\theta$ of $\theta$, so that 
$\hat\theta \sim N( \theta, \sigma^2)$, where $\theta$ and $\sigma^2$ are both unknown.  Suppose additionally that we have an 
estimator $\hat\sigma^2$ of $\sigma^2$ 
for which $\nu \hat\sigma^2/\sigma^2 \sim \chi^2_\nu$, and 
that is independent of $\hat\theta$. 
Then $\hat\theta/\hat\sigma \sim t_{\nu}$ if $\theta=0$. 
Considering the form of the 
FAB $p$-value in the known-variance case, it seems reasonable to 
use as a $p$-value 
\begin{equation}  
 1- |F_\nu( \hat\theta/\hat\sigma + 2 \mu  \tilde \sigma/\tau^2 ) - F_\nu(-\hat\theta/\hat\sigma  )|, 
\label{eqn:ptapprox} 
\end{equation}
where $F_\nu$ is the CDF of the $t_\nu$ distribution, 
$\mu$ and $\tau^2$ are the mean and variance of a distribution 
describing the indirect information about $\theta$, and $\tilde \sigma$ is a prior guess as to the value of $\sigma$. 
This $p$-value is uniformly distributed under the null distribution 
$\hat\theta/\hat\sigma \sim t_\nu$ regardless of the values of 
$\mu,\tau^2$ and $\tilde \sigma$. 
Furthermore, this $p$-value is easy to calculate and is numerically stable
insofar as the CDF of the $t$-distribution may be calculated.  

However, this $p$-value function does not correspond exactly to a 
Bayes-optimal test. 
Letting $T= \hat\theta/\hat\sigma$, a
Bayes-optimal level-$\alpha$ test of 
$H:\theta=0$ is one that rejects when the ratio  $p_\pi( T)/p_0(T)$ exceeds 
its $1-\alpha$ quantile under the null distribution, where 
$p_0$ is the $t_\nu$ density function and $p_\pi$ is the marginal 
density of $T$ under a prior distribution for 
$\theta$ and $\sigma^2$. To see the difficulty in 
using a formally Bayes-optimal test to construct a $p$-value function 
in this case, 
consider a $N(\mu,\tau^2)$ prior distribution for $\theta$ and 
an 
arbitrary 
prior distribution for 
$\sigma^2$. Marginally over the prior distribution for $\theta$ 
but conditional on $\sigma^2$, 
we have $\hat\theta \sim N(\mu ,\sigma^2+\tau^2)$, and the resulting 
marginal distribution for $T=\hat\theta/\hat\sigma$ is that $c T$ has a non-central $t$-distribution with $\nu$ degrees of freedom and  
noncentrality 
parameter $\mu/\sqrt{\sigma^2+\tau^2}$, 
where $c = \sqrt{\sigma^2}/\sqrt{\sigma^2+\tau^2}$.
Let $p_\pi(t|\sigma^2)$ denote the corresponding density for $T$. The marginal density $p_\pi(t)$, needed to construct the Bayes-optimal test,  
is obtained by integrating 
$p_\pi(t|\sigma^2)$ with respect to the prior distribution for $\sigma^2$. 
To obtain the corresponding $p$-value, we would first have to 
compute $p_\pi(t)/p_0(t)$ where $t$ is the observed value of 
$\hat\theta/\hat\sigma$, and then compute the probability of 
obtaining a more extreme value under the null hypothesis, that is, 
$\Pr( p_\pi(T)/p_0(T) > p_\pi(t)/p_0(t) )$ where 
$T\sim t_{\nu}$. 
Calculation of such a $p$-value would be numerically intensive:
It would require two nested levels of numerical integration - one to compute 
$p_\pi(t)$ and one to compute the CDF of $p_\pi(t)/p_0(t)$. Additionally, 
the integrals require evaluation of the non-central $t$-distribution, 
numerical approximation of which is often poor. 
For these reasons, 
the $p$-value function (\ref{eqn:ptapprox})  is recommended unless there is a strong desire 
for a more formal incorporation of specific prior information about $\sigma^2$.

\section{Adaptive $p$-values from independent normal data}

\subsection{FAB $p$-values via indirect information} 

In this section the utility of the 
FAB $p$-value is broadened  by the  development of  statistical methods 
with which indirect information may be summarized and incorporated into 
the $p$-value. 
Generalizing the approach taken by \citet{yu_hoff_2018} and 
\citet{burris_hoff_2018} in the context of confidence interval construction, 
the proposed method is  
to use a statistical model to relate a parameter $\theta$ of interest 
to other data that are available, in which case the indirect information itself 
may be stochastic. This does not affect the uniformity of the FAB $p$-value, 
as long as the summary of the indirect information is statistically independent 
of the direct information. 

\begin{cor}
\label{cor:fabp}
Let $Z$ and $b$ be independent random variables with $Z\sim N(0,1)$. Then $\Pr( 1- |\Phi( Z +  b) - \Phi(-Z) | < u ) =u$.
\end{cor}
\begin{proof}
By independence, the conditional probability given $b$ that  
the $p$-value is less that $u$ is equal to $u$ for each $b$, and 
so it equal to  $u$ marginally over values of $b$ as well. 
\end{proof}

Specifically, we develop adaptive FAB $p$-values in multiparameter 
settings, where indirect information about a given parameter is derived 
from the direct information about other parameters. To start, 
suppose a data vector $\bs Y$ with uncorrelated elements will be sampled from
a multivariate normal population, so that the \emph{sampling model}
for the data is
\begin{equation}
\bs Y\sim N_{p}( \bs\theta , \bs \Sigma),
\label{eqn:ismod}
\end{equation}
where $\bs\theta \in \mathbb R^p$ is unknown and for now we assume $\bs\Sigma = \text{diag}(\sigma_1^2,\ldots, \sigma_p^2)$ is 
known.
For each $j=1,\ldots, p$ we will construct a FAB $p$-value $p(Y_j/\sigma_j,b_j)=  1-| \Phi(Y_j/\sigma_j+ b_j) - \Phi(-Y_j/\sigma_j)|$, 
where $Y_j$ is the $j$th element of $\bs Y$ and $b_j$ is independent 
of $Y_j$, so that $p(Y_j/\sigma_j,b_j)$
is uniformly distributed under the null hypothesis 
$H_j: \theta_j = 0$. 

Indirect information about $\theta_j$ that is independent of $Y_j$  may be derived from
the other entries of $\bs Y$. To facilitate generalization to the correlated 
case considered in the next section, we write these elements 
as $\bl G_j^\top \bs Y$, where $\bl G_j$ is the $p\times (p-1)$ matrix obtained 
by deleting the $j$th column from the $p\times p$ identity matrix. 
Then $\bl G_j^\top \bs Y \sim N_{p-1}( \bl G_j^\top\bs\theta , \bl G_j^\top \bs\Sigma \bl G_j )$,
and $\bl G_j^\top \bs Y$ is statistically independent of $Y_j$. 
By Corollary \ref{cor:fabp}, a  $p$-value $p(Y_j/\sigma_j,b_j)$
is uniformly distributed under $H_j$ for any statistic $b_j$ that is
a function of $\bl G_j^\top \bs Y$. 

While $\bl G_j^\top \bs Y$ doesn't contain direct information about $\theta_j$,
it does contain information about $\bl G_j^\top \bs \theta$, and so $\bl G_j^\top \bs Y$
can be used to provide information about $\theta_j$ if there are relationships
among the entries of $\bs\theta$. Consider a Gaussian \emph{linking model}
to describe these relationships, of the form
\begin{equation} 
\bs \theta\sim N_p(\bs\mu, \bs\Psi). 
\label{eqn:pmod}
\end{equation}
Straightforward linear algebra shows that
the conditional distribution of $\bs\theta$ given $\bl G_j^\top \bs Y$ is $N_p(\bl m,\bl V)$
where
\begin{align}
\bl V & = [ \bs\Psi^{-1} +\bl  G_j ( \bl G_j^\top \bs\Sigma\bl  G_j )^{-1}\bl  G_j^\top ]^{-1} \label{eqn:cvar} \\
\bl m & =\bl  V  [ \bs \Psi^{-1}\bs\mu + \bl G_j ( \bl G_j^\top \bs\Sigma \bl G_j )^{-1} \bl G_j^\top \bs Y ] \label{eqn:cexp} . 
\end{align}
In some cases 
it may be easier to compute these as 
$\bl V  = \bs \Psi -\bs \Psi\bl  G_j [\bl  G_j^\top(\bs\Psi+\bs\Sigma)\bl G_j ]^{-1} \bl G_j^\top\bs \Psi$ and 
$\bl m  =\bs \mu +\bs \Psi \bl G_j [ \bl G_j^\top(\bs\Psi+\bs\Sigma)\bl G_j ]^{-1} \bl G_j^\top (\bs Y -\bs \mu ).$

To the extent that the linking model is believed, 
the conditional distribution of $\theta_j$ given 
 $\bl G_j^\top \bs Y$ 
is therefore 
$\theta_j| \bl G_j^\top \bs Y \sim N(m_j,v_{j,j})$, 
where $m_j$ and $v_{j,j}$ denote elements $j$ and $\{j,j\}$ of 
$\bl m$ and $\bl V$ respectively. 
This distribution quantifies the indirect information about
$\theta_j$: It
is the
distribution of $\theta_j$ conditional on the data $\bl G_j^\top \bs Y$ 
that is independent of $Y_j$ under the sampling model (\ref{eqn:ismod}), 
which in turn provides indirect 
information  about $\theta_j$ via the 
relationships among the elements of $\bs\theta$
based on the
linking model (\ref{eqn:pmod}). 

Even if the linking model is not believed, it may still be used as a way 
to construct a data-adaptive $p$-value that will be uniformly distributed 
if $H_j$ is true, as long as $m_j$ and $v_{j,j}$ are statistically 
independent of $Y_j$ under the sampling model (\ref{eqn:ismod}). 
The performance of such a $p$-value under alternatives to $H_j$ 
will depend on the extent to which the linking model 
is a good representation of the heterogeneity of the elements 
of $\bs\theta$. 
The accuracy of this representation may be improved by estimating 
$\bs\mu$ and $\bs\Psi$ from the data itself. 
In the applications we consider, $\bs\mu$ and $\bs\Psi$ will depend 
on some lower dimensional parameter, say $\bs\gamma$.  
This parameter may be estimated using hierarchical modeling methods, 
as the sampling model for the indirect information, $\bl G_j^\top \bs Y\sim N_{p-1}( \bl G_j^\top \bs\theta, \bl G_j^\top\bs\Sigma \bl G_j )$,  together with the induced linking model $\bl G_j^\top \bs\theta \sim N_{p-1}( \bl G_j^\top \bs\mu_{\bs\gamma},\bl  G_j^\top \bs\Psi_{\bs\gamma} \bl G_j )$  constitute a Gaussian mixed effects model. 
For example, the maximum likelihood estimator of $\bs\gamma$ 
based on data independent of $Y_j$ 
is the maximizer 
in $\bs\gamma$ of the density of the 
marginal distribution of $\bl G_j^\top \bs Y$, 
$\bl G_j^\top \bs Y \sim N_{p-1}(\bl  G_j^\top \bs \mu_{\bs\gamma},\bl  G_j^\top [\bs \Sigma + \bs \Psi_{\bs\gamma} ] \bl G_j)$.

Now we are in a position to apply the results of the previous
section. By Theorem \ref{thm:fabtest} the Bayes-optimal
test of $H_j$ given the indirect information $\bl G_j^\top \bs Y$ rejects when
$|Y_j/\sigma_j + m_j \sigma_{j}/v_{j,j}|$ is large,
and by Theorem \ref{thm:fabpvalue}
the Bayes-optimal $p$-value is given by
$p(Y_j/\sigma_j,b_j)$
where $b_j = 2 m_j \sigma_{j}/v_{j,j}$. 
However, as $m_j$ and $v_{j,j}$ depend on the unknown 
parameter $\bs\gamma$ via  $\bs\mu_{\bs\gamma}$ and $\bs\Psi_{\bs\gamma}$, we instead use 
estimates $\tilde m_j$ and $\tilde v_{j,j}$ based on 
$\bs \mu_{\tilde {\bs \gamma}}$ and $\bs\Psi_{\tilde {\bs \gamma}}$, where 
$\tilde{\bs \gamma}$ is obtained from $\bl G_j^\top \bs Y$. Since $\bl G_j^\top \bs Y$ 
is independent of 
$Y_j$,   $\tilde b_j = 2 \tilde m_j \sigma_j/\tilde v_{j,j}$ is 
also independent of $Y_j$ and 
so $1-|\Phi(Y_j/\sigma_j+ \tilde b_j ) -\Phi(-Y_j/\sigma_j) |$ is uniformly distributed 
if $\theta_j=0$ by Corollary \ref{cor:fabp}. 

If $\sigma_j^2$ is not known
but a high-precision estimator is available, then 
the value of the estimator can be plugged into the formula for the FAB
 $p$-value. 
Alternatively, if an estimator $\hat\sigma^2_j$ is available 
that is independent of $Y_j$ and for which 
$\nu \hat\sigma^2_j/\sigma_j^2\sim \chi^2_\nu$, then by  
Theorem \ref{thm:gfabp} 
the $p$-value $1- |F_\nu( Y_j/\hat\sigma_j + \tilde b_j ) - F_{\nu}(-Y_j/\hat\sigma_j) |$ is uniformly 
distributed if $\theta_j=0$ and $\tilde b_j$ is independent of 
 $Y_j/\hat \sigma_j$, 
where $F_{\nu}$ is the CDF of the $t$-distribution with $\nu$ degrees of freedom. 
Based on the discussion at the end of Section 2, we use 
$\tilde b_j = 2 \tilde m_j \tilde \sigma_{j}/\tilde v_{j,j}$, 
where $\tilde m_j$ and $\tilde v_{j,j}$ are estimated 
from $\bl G_j^\top \bs Y$ as described in the previous paragraph. 
Note that 
an additional estimator $\tilde\sigma_j$ of $\sigma_j$ is also required, 
which must be independent 
of $Y_j/\hat\sigma_j$ for exact uniformity of the $p$-value 
to be maintained when $\theta_j=0$. Availability of such an estimator 
will depend on the particular application. For example, if estimates 
of the $\sigma_k$'s for which $k\neq j$ are also available, then $\tilde \sigma_j^2$ 
may be based on them.

\subsection{Example: Small area inference with the Fay-Herriot model} 
The 2002 Educational Longitudinal Study gathered data on a sample of
U.S.\ high schools and their students. 
From each
 participating school, a small sample of 10th grade
students
were selected and given a survey and a standardized reading exam. 
For $p=684$ schools the sample size was 2 students or more, 
and among these the median school-level sample size was 21 and the maximum 
sample size was 50. 
In this section we use adaptive 
FAB $p$-values to evaluate for each school 
the evidence that their school-specific 
mean score on the reading exam deviates from a 
particular national average value.

Let $\bar Y_j$ and $\hat\sigma^2_j$ be the sample mean and
variance of the $n_j$ reading test scores of students sampled from school $j$.
A Gaussian within-school sampling model implies that 
the vector $\bar{\bs Y}$ of school-specific sample means is normally 
distributed, with 
\begin{equation} \bar {\bs Y} \sim N_p(\bs\theta, \text{diag}(\sigma^2_1/n_1,\ldots, 
\sigma^2_p/n_p)),
\label{eqn:sasampling} 
\end{equation} where $\bs\theta \in \mathbb R^p$ is the vector 
of ``true'' school-specific means, meaning that $\theta_j$ is 
the average exam score 
had all students in school $j$ participated in the study. 
The sampling model also implies that 
$(n_j-1)\hat\sigma^2_j/\sigma_j^2 \sim \chi^2_{n_j-1}$, independently for $j=1,\ldots, p$ and independent of $\bar {\bs Y}$.   

In this numerical illustration we construct $p$-values for evaluating 
the hypothesis that 
$\theta_j$ is equal to 50, 
a score 
that corresponds to the intended national average on the exam.
Assuming the Gaussian within-school sampling model, the statistic
$T_j= \sqrt{n_j}(\bar Y_j - 50)/\hat\sigma_j$ has a $t_{n_j-1}$ distribution under the null hypothesis $H_j:\theta_j=50$.
Letting $\nu_j=n_j-1$, the $p$-value
$1- |F_{\nu_j}(T_j + \tilde b_j ) - F_{\nu_j}(-T_j)|$
will be uniformly distributed under the null hypothesis as long as 
$\tilde b_j$ is independent of $T_j$.  
The usual UMPU $p$-value based on $T_j$ is obtained by setting 
$\tilde b_j=0$. An  adaptive 
FAB $p$-value for each school $j$ can be constructed by setting  
$\tilde b_j$ to 
a value that utilizes indirect information from schools other than 
$j$. We do this via a Gaussian linking model for the $\theta_j$'s 
of the form 
\begin{equation} \bs\theta\sim N_p (\bl X\bs\beta, \tau^2\bl  I),  
\label{eqn:salinking} 
\end{equation}
where 
$\bl X$ is a $p\times q$ matrix of observed school-level characteristics, 
and $\bs\beta$ and $\tau^2$ are unknown. 
For these data, the characteristics $\bl x_j$ of school $j$ includes a numeric 
measure of the number of students in school $j$ who participated in a 
free lunch program, the total enrollment of school $j$,  and seven indicator variables that encode three categorical 
variables: school type (public, Catholic, or other private), 
region of the country (West, Midwest, South and East), and 
urbanicity (urban, suburban, rural). 
Together, a  sampling model such as  (\ref{eqn:sasampling}) and
a linking model such as  (\ref{eqn:salinking}) are sometimes referred to as
a Fay-Herriot model  \citep{fay_herriot_1979}, a linear mixed-effects model that is often used in the small area estimation literature to
obtain stable estimates
for each of many groups or ``areas''  by
sharing information across groups \citep{ghosh_rao_1994}.
Here  we are only using the linking model 
to obtain a value of $\tilde b_j$ with which the FAB $p$-value 
for group $j$ is constructed. 
Validity of the linking model is not necessary for the 
FAB $p$-value to be uniformly distributed under the null hypothesis $H_j$.

\begin{figure}
\begin{center}
\begin{knitrout}\footnotesize
\definecolor{shadecolor}{rgb}{0.969, 0.969, 0.969}\color{fgcolor}

{\centering \includegraphics[width=6.25in]{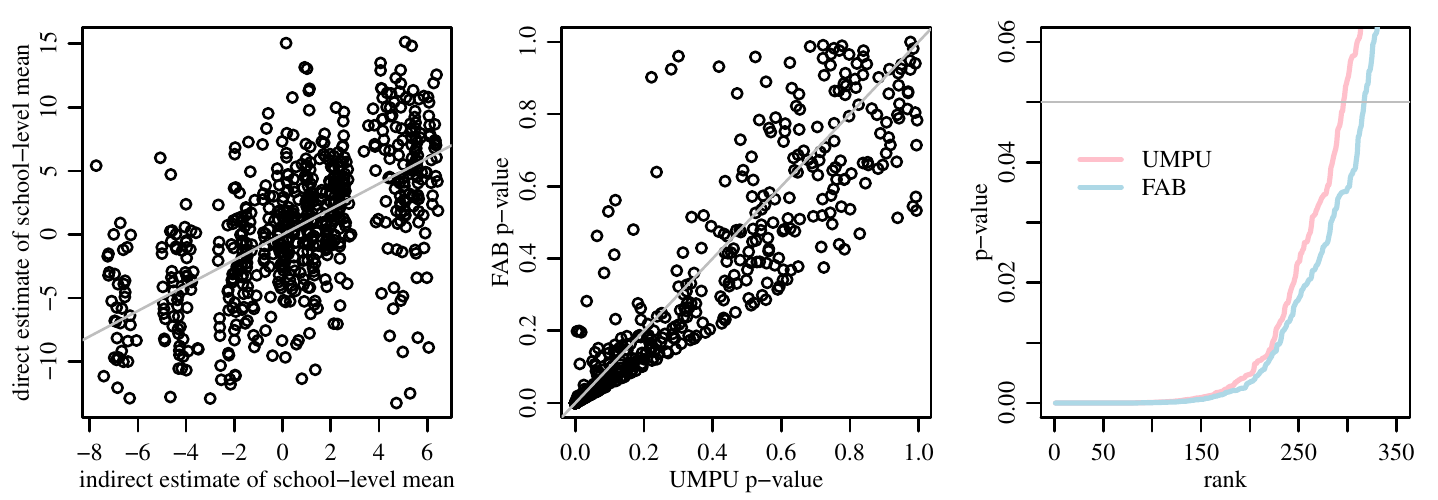} 

}

\end{knitrout}
\caption{Results from the ELS analysis. The left 
 panel plots direct estimates of school-level means versus the indirect 
 estimates.  
The middle and right panels plot the FAB and UMPU $p$-values. }
\label{fig:elsFH} 
\end{center}
\end{figure}

Based on the discussion in Section 3.1, we set
$\tilde b_j=2(\tilde {\bs\beta}^\top \bl x_j - 50 ) (\tilde \sigma/\sqrt{n_j})/\tilde \tau^2$,
where $\tilde {\bs \beta}$, $\tilde \tau^2$ and $\tilde \sigma^2$ are the
MLEs obtained by fitting the the Fay-Herriot hierarchical model
(with a common within-school variance) to the data from schools
other than $j$.
Since $\tilde b_j$ is statistically independent of the sample from
school $j$, the FAB 
$p$-value $1- |F_{\nu_j}(T_j+ \tilde b_j) - F_{\nu_j}(-T_j)|$
is  uniformly distributed if $\theta_j=50$,
even if the linking model is incorrect.
Figure \ref{fig:elsFH} 
describes some aspects of this data analysis. The left-most panel 
plots the 
direct estimate $\bar Y_j$ of each $\theta_j$ versus the 
indirect estimate $\tilde m_j$ obtained from 
the Fay-Herriot model.
The positive correlation between these two estimates suggests that 
the indirect information can be of use for school-level inference. 
The middle panel of the figure plots the FAB $p$-values versus the 
UMPU $p$-values corresponding to UMPU $t$-tests applied to each 
school individually. 
The FAB $p$-value is smaller than the UMPU $p$-value for 529  
out of 
684  schools (77\%).  
Of more interest might be 
the small $p$-values, displayed in the right-most panel of the figure.  
There were 316 FAB $p$-values and 295 UMPU $p$-values less than 0.05, 
for a difference of 21 schools. 

\subsection{Example: Spatial linking models}
In this subsection we perform a small numerical simulation study that is meant to mimic the search 
for signals along a one-dimensional lattice, such as a chromosome. 
Let $\bs Y \sim N_p( \bs \theta ,\bl  I)$
and suppose the true $\bs\theta$  is a realization of a 
discrete Markov chain
taking values in $\{ -1,0,1\}$ and with transition probability matrix
\[
P= \begin{pmatrix} .975 & .025 & .000 \\
                   .010 & .990 & .010 \\
                   .000 & .025 & .975 \end{pmatrix}. 
\]
Figure \ref{fig:indmm} displays one realization of such a process for $p=1000$.

\begin{figure}[ht]
\begin{knitrout}\footnotesize
\definecolor{shadecolor}{rgb}{0.969, 0.969, 0.969}\color{fgcolor}

{\centering \includegraphics[width=6in]{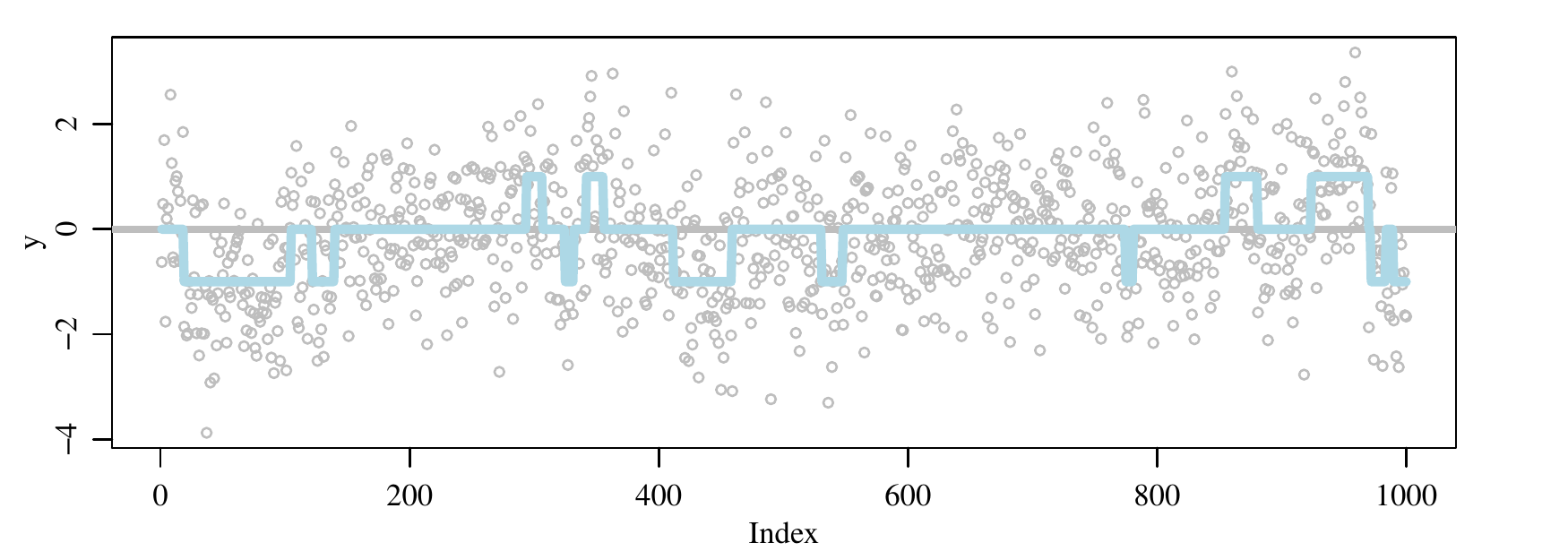} 

}

\end{knitrout}
\caption{A realization of $\bs\theta$ (blue line) and $\bs Y$ (gray dots) 
         from the hidden Markov model. } 
\label{fig:indmm}
\end{figure}

For each $j\in \{1,\ldots, p\}$
we construct a FAB $p$-value from $\bs Y$  to 
evaluate $H_j:\theta_j=0$ 
using the (incorrect) linking model $\bs\theta\sim N_p(\mu \bl 1,\bs\Psi)$ where 
$\mu$ is a scalar, $\bl 1$ is a $p$-dimensional vector of ones, and 
$\bs\Psi$ is the covariance matrix of a conditional autoregressive process, 
so that 
the conditional distribution of $\theta_j$ given the other elements of 
$\bs\theta$ is 
$N( \beta_0 + \beta_1 ( \theta_{j-1} + \theta_{j+1}) , \tau^2)$ 
for some values of $\beta_0$, $\beta_1$ and $\tau^2$. 
To construct a FAB $p$-value for $H_j$, the procedure is to 
\begin{enumerate}
\item obtain estimates $\tilde {\bs\gamma}=(\tilde \mu,\tilde \beta_0,\tilde \beta_1,\tilde \tau^2)$  of the parameters $\bs\gamma=(\mu, \beta_0,\beta_1,\tau^2)$ 
of the  Gaussian autoregressive linking model, using 
 $\bl G_j^\top\bs Y$ as data; 
\item compute plug-in estimates $\tilde m_j, \tilde v_{j,j}$ of the conditional mean and variance $m_j, v_{j,j}$ of $\theta_j$ given $\bl G_j^\top \bs Y$, using $\tilde {\bs \gamma}$ and Equations \ref{eqn:cvar} and 
\ref{eqn:cexp}; 
\item compute the FAB $p$-value 
$p(Y_j,\tilde b_j) = 1-| \Phi( Y_j + \tilde b_j) - \Phi(-Y_j)|$, where $\tilde b_j = 2 \tilde m_j/\tilde v_{j,j}$. 
\end{enumerate}
Further computational details 
are available in the supplementary material for this article.

Computing the FAB $p$-value for each of the 1000 hypotheses is somewhat time consuming, as 
the autoregressive model in step 1 is fit separately for each $j\in \{1,\ldots, 1000\}$, 
each time without $Y_j$
to ensure that $\tilde b_j$ is independent of $Y_j$. 
One faster alternative is to 
break the vector $\bs Y$ into contiguous subvectors, and 
then construct the FAB $p$-values 
in one subvector using linking model parameters estimated from the 
remainder of the 
subvector. Such 
an approach will maintain the uniformity of the FAB $p$-values under 
$H_j$, but will not make optimal use of the information in 
$\bl G_j^\top \bs Y$. 
Another option would be to
fit the autoregressive model once using the entire $\bs Y$-vector, and use the resulting 
linking model parameter estimate $\tilde {\bs\gamma}$
 to construct each $p$-value. 
However, since in 
this case $Y_j$ is used to obtain the common $\tilde {\bs\gamma}$, the value of 
$\tilde b_j$ will depend slightly on $Y_j$,
violating the 
sufficient condition for $p(Y_j,\tilde b_j)$ to be uniformly 
distributed under $H_j$. 
However, as this 
linking model  has a small number of parameters relative to $p$, 
we might expect the influence of $Y_j$ on $\tilde b_j$ to be slight.
We investigate this claim numerically in the simulation study  that follows.

One-hundred $\bs\theta$-vectors were simulated from the discrete Markov chain described above, and from each a single observed data vector $\bs Y\sim N_p(\bs\theta, \bl I)$ was simulated. For each $\bs Y$-vector and each index $j$, 
three $p$-values 
for evaluating $H_j:\theta_j=0$ were computed:  the adaptive FAB $p$-value described in steps 1, 2 and 3 above; an ``approximate'' adaptive 
FAB $p$-value that 
estimates the parameters of the linking model only once for each simulated $\bs Y$-vector;  and the usual $p$-value $p(Y_j,0)$ based on the UMPU test of 
$H_j:\theta_j=0$. 
For each simulated data vector and each type of $p$-value, 
we computed the number of 
discoveries (rejected null hypotheses)
and the false discovery proportion (FDP)
using the Benjiamini-Hochberg procedure \citep{benjamini_hochberg_1995} with a false discovery rate (FDR)
controlled to be less than $0.2$.

We first compare each of the three $p$-value 
procedures in terms of FDR control. 
On average across datasets, the number of ``true nulls'' ($\theta_j$'s equal to zero) was about 574 out of 1000.
Therefore, the Benjamini-Hochberg (BH)
 procedure using the UMPU $p$-values should 
attain an actual FDR of about $0.574\times 0.2 \approx 0.115$. 
Modulo Monte Carlo error, this is what was observed - the FDP
based on the UMPU $p$-values was  0.106
on average across the 100 datasets.
The average FDP using the
FAB and approximate FAB procedures were both about  0.108, 
very similar to that of the UMPU procedure and well below the target
FDR. We note that 
control using the 
Benjamini-Hochberg (BH) procedure is guaranteed for the 
UMPU $p$-values since they are statistically independent. 
In contrast, the FAB $p$-values are dependent 
since, for example,  the $p$-values for $\theta_j$ and $\theta_{j+1}$ 
are both functions of $Y_j$ and $Y_{j+1}$ (the indirect information 
for $Y_j$ includes $Y_{j+1}$ and vice versa). This leads to 
positive dependence among spatially proximal FAB $p$-values.  
However, several theoretical results have been obtained showing that 
for some types of positive  dependence
 the BH procedure maintains the target FDR 
\citep{benjamini_yekutieli_2001,sarkar_2002}, or does so asymptotically 
in $p$ \citep{clarke_hall_2009}. While the conditions for these results 
are not met exactly for the example scenario presented here, 
these theoretical results and the empirical results for this example
(that the actual FDP in the example is less than half the target FDR)
suggest that a bigger issue is the 
conservatism of the BH procedure.

\begin{figure}[ht]
\begin{knitrout}\footnotesize
\definecolor{shadecolor}{rgb}{0.969, 0.969, 0.969}\color{fgcolor}

{\centering \includegraphics[width=5.75in]{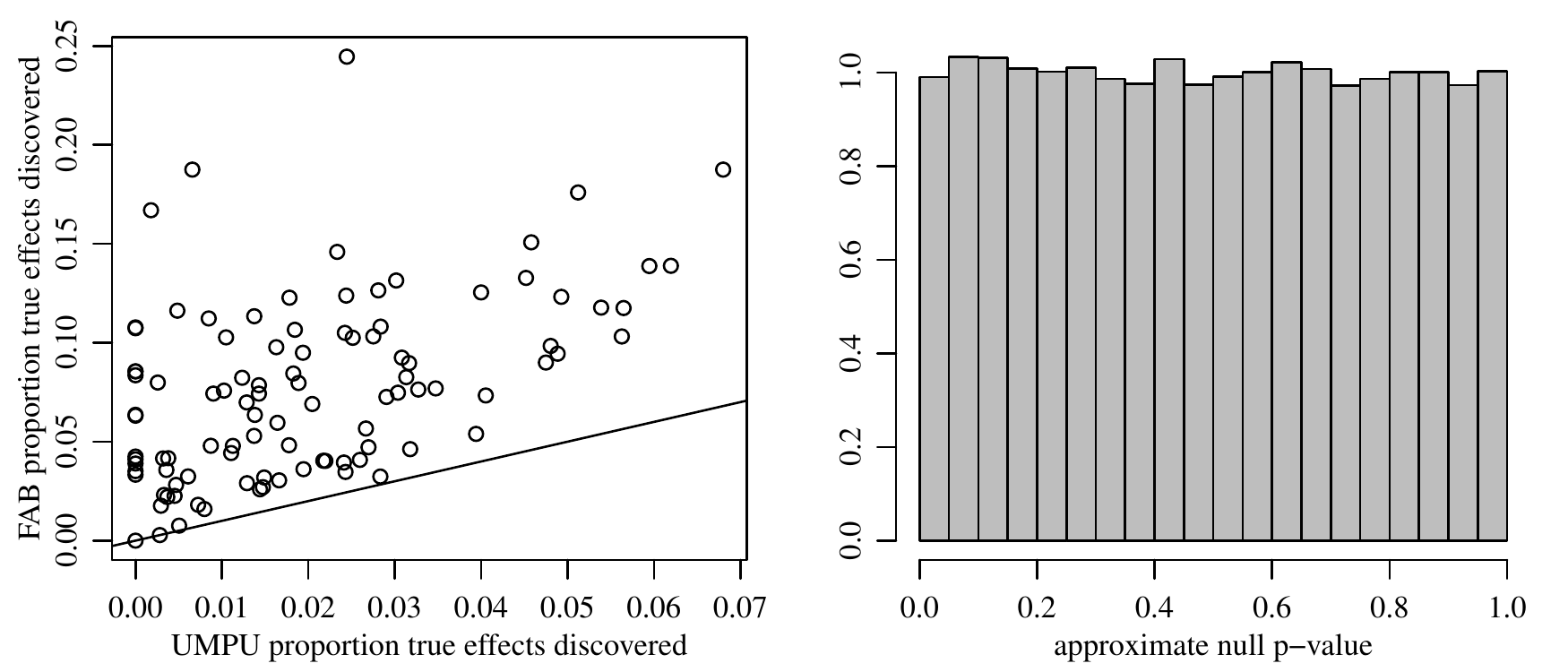} 

}

\end{knitrout}
\caption{Some results from the simulation study. The panel on the left plots 
 the fraction of true positives discovered by the FAB procedure relative 
to the UMPU procedure for each of the 1000 simulated datasets. The panel 
on the right displays the distribution of the approximate FAB $p$-values 
corresponding to true nulls. }  
\label{fig:spatialresults}
\end{figure}

Some of the results relating to power are displayed in
Figure \ref{fig:spatialresults}. The first panel plots the
fraction of true positives discovered using the exact FAB procedure,
versus the fraction discovered by the UMPU procedure.
The proportion discovered by the FAB approach was as high or higher 
for all 1000
simulated datasets. On average across datasets, the mean proportion
of true positives that were discovered was 0.02 
for the UMPU procedure
and 0.077 for the FAB procedure, almost a four-fold improvement.

Finally, we note that the approximate FAB procedure
performed nearly identically to the exact FAB procedure. Recall that
the only difference between the two procedures is that for each $j$,
the
parameter $\bs\gamma$ of the linking model was estimated
using all elements of $\bs Y$ for the approximate procedure, and all
but the $j$th element of $\bs Y$ for the exact procedure. Since $(p-1)/p$
is quite close to one in this case, the two estimators are nearly the
same and we would therefore expect
the $p$-values to be very similar. In particular, since the
FAB $p$-values are uniformly distributed for $j$'s such
that $\theta_j=0$, we expect the corresponding approximate 
FAB $p$-values to 
be nearly uniformly distributed as well.
This is confirmed in the second panel of Figure \ref{fig:spatialresults},
which plots a histogram of the approximate FAB $p$-values for the 
true null hypotheses, across all
iterations of the simulation study.

\section{Adaptive $p$-values from dependent data} 

\subsection{Data with known covariance}
In many applications the elements of the data vector 
$\bs Y$ will be correlated. Examples include the case that 
$\bs Y$ is a vector of estimates of linear regression coefficients from 
a non-orthogonal design matrix, and situations where 
the elements of $\bs Y$ correspond to 
measurements that are spatially or temporally related. 
In such cases, FAB $p$-values that are marginally uniform under their 
null hypotheses may be obtained as in the previous section
by generalizing the form of the indirect information used to 
construct the $p$-value for each hypothesis. 

Specifically, consider constructing a $p$-value for 
$H_j:\theta_j=0$ based on the model 
$\bs Y\sim N_p(\bs\theta, \bs\Sigma)$ where $\bs\Sigma$ is not necessarily diagonal.  In this subsection we consider the case that $\bs\Sigma$ is known - 
cases where it is unknown are considered in the following two subsections. 
As before, we refer to  $Y_j/\sigma_j$ as the direct information 
for $\theta_j$. 
Indirect information about $\theta_j$ that is independent of $Y_j$  may be based upon an
appropriate linear transformation of $\bs Y$.
Let 
$\bl G_j$ be any $p\times (p-1)$ matrix
whose columns form a basis for the null space of the
$j$th column of $\bs\Sigma$.
Then  $\bl G_j^\top \bs Y \sim N_{p-1}( \bl G_j^\top\bs\theta , \bl G_j^\top \bs\Sigma \bl G_j )$,
and  is statistically independent of $Y_j$.
By Corollary \ref{cor:fabp}, a  $p$-value $p(Y_j/\sigma_j,\tilde b_j) =  1- | \Phi(Y_j/\sigma_j+ \tilde b_j ) -\Phi(-Y_j/\sigma_j) |$
is uniformly distributed under $H_j$ for any statistic $\tilde b_j$ that is
a function of $\bl G_j^\top \bs Y$.

As before, while $\bl G_j^\top \bs Y$ doesn't provide direct information about $\theta_j$, 
it does provide information about $\bl G_j^\top \bs\theta$ 
and so it 
can be used to provide information about $\theta_j$ if there are relationships
among the entries of $\bs\theta$. 
If such relationships can be encoded in a linking model 
$\bs\theta \sim N_p(\bs\mu, \bs\Psi)$, then 
a FAB $p$-value for $H_j$ may be constructed exactly as in the previous section, modulo the more general form of the indirect information 
$\bl G_j^\top \bs Y$ where $\bl G_j$ depends on $\bs\Sigma$.   
In particular, 
the conditional expectation and variance of $\theta_j$ given 
$\bl G_j^\top \bs Y$ under the linking model may still be 
obtained from Equations \ref{eqn:cvar}  and \ref{eqn:cexp}. 
Additionally, if reasonable values of  $\bs \mu$  and $\bs\Sigma$ are not known 
in advance, then they may be estimated from the indirect 
information $\bl G_j^\top \bs Y$ and the resulting FAB $p$-value will still 
be uniform under $H_j$. 

In practice it is 
rare that the covariance of $\bs Y$ will be completely known. 
However, if it is known sufficiently to decorrelate $\bs Y$ 
then FAB $p$-values may be obtained that are exactly uniform 
under their null hypotheses. Otherwise, if $\bs\Sigma$ may be consistently 
estimated, then FAB  $p$-values may be constructed that 
are asymptotically uniform. In this case, we can also replace the 
normality assumption for $\bs Y$ with asymptotic normality. 
We illustrate procedures 
for these two scenarios in the following 
subsections.  

\subsection{Data with known correlation} 
Here we show how to construct FAB $p$-values for regression coefficients of a normal linear regression model, a case where the correlations among the direct estimates are known. Construction of
FAB $p$-values in other cases of known correlation is similar.

We use essentially the same method of data splitting as was used 
in \citet{hoff_yu_2019} in the context of confidence interval construction. 
Suppose we wish to construct $p$-values for the elements of $\bs\beta\in \mathbb R^p$ 
in the linear regression model $\bs Y \sim N_n(\bl X\bs\beta,\sigma^2 I)$, 
where $\bl X\in \mathbb R^{n\times p}$ is an observed design matrix 
and $\bs\beta\in \mathbb R^p$ and $\sigma^2\in \mathbb R^+$ are unknown. Our sampling model is based 
on the probability distribution of the ordinary least-squares (OLS) regression estimator, 
$\hat{\bs \beta} \sim N_p(\bs\beta ,\sigma^2  (\bl X^\top\bl X)^{-1} )$, 
with $\hat{\bs\beta}$ being independent of the usual unbiased estimator $\hat\sigma^2$ of $\sigma^2$, for which 
$\nu \hat \sigma^2/\sigma^2 \sim \chi^2_\nu$, where $\nu=n-p$. 
The direct estimate for a particular coefficient $\beta_j$ is 
$\hat\beta_j \sim N(\beta_j, h_{j,j} \sigma^2 )$ where 
$h_{j,j} = (\bl X^\top \bl X)^{-1}_{j,j}$, and  so 
$T_j=\hat \beta_j/[h_{j,j}^{1/2} \hat\sigma]\sim t_{\nu}$ 
under $H_j:\beta_j=0$.  
To find the indirect information, 
let $\bl G_j$ be a $p\times (p-1)$ matrix 
whose columns span the null space of the $j$th column of 
$(\bl X^\top \bl X)^{-1}$, so that  $\bl G_j^\top \hat {\bs \beta}$ 
is independent of 
$\hat\beta_j$. 
Letting $F_\nu$ denote the CDF of a $t$-distribution with $\nu$ degrees of 
freedom, 
the distribution of $1 - |F_\nu(T_j + \tilde b_j) - F_\nu(-T_j)|$ will then 
be uniformly distributed under $H_j:\beta_j=0$ for any scalar 
function $\tilde b_j$ of $\bl G_j^\top \hat {\bs\beta}$.

A value of $\tilde b_j$ may be obtained from a linking model for $\bs\beta$. 
The marginal distribution of $\bl G_j^\top \hat{\bs \beta}$ under the linking model 
$\bs\beta\sim N(\bs\mu,\bs\Psi)$ is
$\bl G_j^\top \hat{\bs \beta} \sim N_{p-1}( \bl G_j^\top \bs\mu, \bl G_j^\top [ \bs\Psi + \sigma^2 (\bl X^\top \bl X)^{-1} ] \bl G_j)$. 
If $\bs \mu$ and $\bl \Psi$ are sufficiently structured 
then marginal maximum likelihood estimates 
$\tilde {\bs \mu}, \tilde {\bs \Psi}, \tilde \sigma^2$ 
may be obtained from this 
marginal  model for $\bl G_j^\top \hat{\bs \beta}$. 
These estimates may
then 
be used to obtain estimates $\tilde m_j$ and $\tilde v_{j,j}$ of 
the conditional mean $m_j$ and variance $v_{j,j}$ of $\beta_j$ given the indirect information. 
Since $\tilde m_j, \tilde v_{j,j}$ and $\tilde \sigma^2$ 
are obtained from $\bl G_j^\top\hat{ \bs\beta}$, they are independent of 
$\hat \beta_j$ (and $\hat\sigma^2$), and so the 
FAB $p$-value $1- |F(T_j+ 2 \tilde m_{j} \tilde\sigma /\tilde v_{j,j}  ) - F(-T_j)| $ 
is uniformly distributed under $H_j$. 

This methodology can also be applied in cases where there is only a 
linking model for a subset of parameters. For example, consider the 
linear regression model 
\[
  \bs Y \sim N_n( \bl W \bs\alpha +\bl  X \bs\beta , \sigma^2 \bl I ) 
\]
where $\bl W$ and $\bl X$ are observed design matrices. If interest is primarily in 
$\bs\beta$ and a reasonable linking model relating $\bs\alpha$ to $\bs\beta$ is
not available, 
then FAB $p$-values for $\bs\beta$ alone may be obtained as just 
described, except noting that 
$\hat{\bs\beta} \sim N_p(\bs\beta ,\sigma^2  \bs\Omega  )$ where 
$\bs\Omega$ is the appropriate submatrix of $ ( [ \bl W \bl X]^\top  [ \bl W \bl X] )^{-1}$. 
Letting $\bl G_j$ be an orthogonal basis for the null space of 
the $j$th column of $\bs\Omega$, we can obtain an estimated conditional 
distribution for $\beta_j$ using a linking model $\bs\beta\sim N_p(\bs \mu ,\bs\Psi)$ and the indirect information $\bl G_j^\top \hat{\bs \beta}$. 
Specifically, $\bl G_j^\top \hat{\bs \beta}$ is independent of $\hat\beta_j$, 
and so the indirect information about $\beta_j$ is then given by 
$\bl G_j^\top \hat{\bs \beta} \sim N_{p-1} ( \bl G_j^\top \bs \mu , \bl G_j^\top [  \bs\Psi + \sigma^2 \bs\Omega]  \bl G_j)$.

As a numerical illustration we perform an analysis of another aspect of the 
2002 Educational Longitudinal Study  dataset. 
Let $Y_i$ be the test score of student $i$, 
$s_i$ be a numerical measure of their socioeconomic status (SES)
and let $\bl v_i$ be a three-dimensional vector indicating the 
sex, parents' education, and status as a native English speaker of 
student $i$. 
Also, let 
$\bl g_i\in \{ 0,1\}^p$ be the binary vector indicating in which of 
$p =  684$ schools student $i$ is enrolled. Let $\bl w_i = (\bl g_i, \bl v_i) \in 
\mathbb R^{p+3}$ and let $\bl x_i = s_i \bl g_i  \in \mathbb R^{p}$. The linear 
regression model 
\[
  Y_i = \bs\alpha^\top \bl w_i + \bs\beta^\top \bl x_i+ \epsilon_i
\]
can be used to evaluate the evidence for each school that 
SES is related to test score, controlling for effects 
of $\bl v_i$, 
by obtaining $p$-values for each element of $\bs\beta$.  
Note that the entries of $\hat{\bs\beta}$, the OLS estimator of $\bs\beta$, 
are correlated with each other because of the 
presence of the covariate $\bl v_i$. 

We compute indirect information for testing each $\beta_j$ using 
$\bl G_j^\top \hat {\bs \beta}$ and 
an 
exchangeable linking model $\bs\beta\sim N_p(\mu \bl 1 ,  \tau^2 \bl I)$. 
The model for the indirect information is thus 
$\bl G_j^\top \hat{\bs\beta}  \sim N_{p-1}( \mu \bl G_j^\top \bl 1 , \sigma^2 \bl G_j^\top \bs\Omega \bl G_j +  \tau^2 \bl I)$. 
From this model 
we obtain marginal maximum likelihood estimates $(\tilde \mu, \tilde \tau^2,\tilde \sigma^2)$ of 
 $(\mu, \tau^2,\sigma^2)$ that are statistically independent of $T_j$. 
We then compute the 
 FAB $p$-value
$1- | F( T_j + 2 \tilde \mu_j \tilde \sigma/\tilde \tau ) - F(-T_j)|$, which is
uniformly distributed if $\beta_j=0$.

Some results of this analysis are plotted in Figure \ref{fig:elsINT}. 
The first panel plots the direct OLS estimate $\hat\beta_j$  for 
each school $j$
versus its estimated  precision $(\hat \sigma^2 \omega_{j,j})^{-1/2}$ 
(the reciprocal of the standard error). 
While  the  variability of the $\hat\beta_j$'s around their average value 
(given by a horizontal black line) is large, 
the variability is highest 
among estimates corresponding to schools with low sample sizes 
(and hence low precision). This suggests that much of the observed 
variability 
among the elements of $\hat{\bs\beta}$ is due to within-school 
sampling variability, and that the heterogeneity of the actual
$\beta_j$'s is much lower. 
An informal assessment of this is made with the two gray curves in the plot, 
which are at the across-school average of the $\hat\beta_j$'s plus and minus 
1.96 times the standard error (the reciprocal of the horizontal coordinate of the plot). Indeed, most of the OLS estimates fall within the standard errors of the overall average. 
Nevertheless, 
an  $F$-test of the  global null hypotheses of 
no across-school variation in $\beta_j$'s has a $p$-value of 
less that $10^{-4}$, suggesting non-zero  across-school variation. 

The middle panel of  Figure \ref{fig:elsINT} plots the UMPU $p$-values 
versus the FAB $p$-values for evaluating $H_j:\beta_j=0$ for each school $j$. 
The unfamiliar pattern is due to the fact that, 
while there is evidence of across-school variability, 
the variation is centered around a positive 
value that is relatively far from zero. 
As a result, the FAB $p$-values have adaptively 
become nearly one-sided, with values close to half those of the 
UMPU $p$-values for the $\hat\beta_j$'s that are positive. 
This leads to many more ``small'' FAB $p$-values than small UMPU $p$-values, 
as shown in the right-most plot in the figure. 
For example, there are 245 FAB $p$-values less than 0.05, but only 188
UMPU $p$-values below this level. 

There may be some concern that since the FAB $p$-values in this case are
nearly the same as one-sided $p$-values, their power to 
detect evidence that a given parameter 
is less than zero is drastically reduced. While this is true, there is 
not much evidence in the data that any of the $\beta_j$'s are actually 
less than zero. 
For example, there was only one  $p$-value smaller than 0.05 that 
corresponded to a negative value of $\hat\beta_j$ 
(a $p$-value of 0.04). 
The data suggest that the $\beta_j$'s have a small amount of variation 
around a positive value, a data feature to which the FAB $p$-value procedure has adapted. 

\begin{figure}
\begin{center}
\begin{knitrout}\footnotesize
\definecolor{shadecolor}{rgb}{0.969, 0.969, 0.969}\color{fgcolor}

{\centering \includegraphics[width=6.5in]{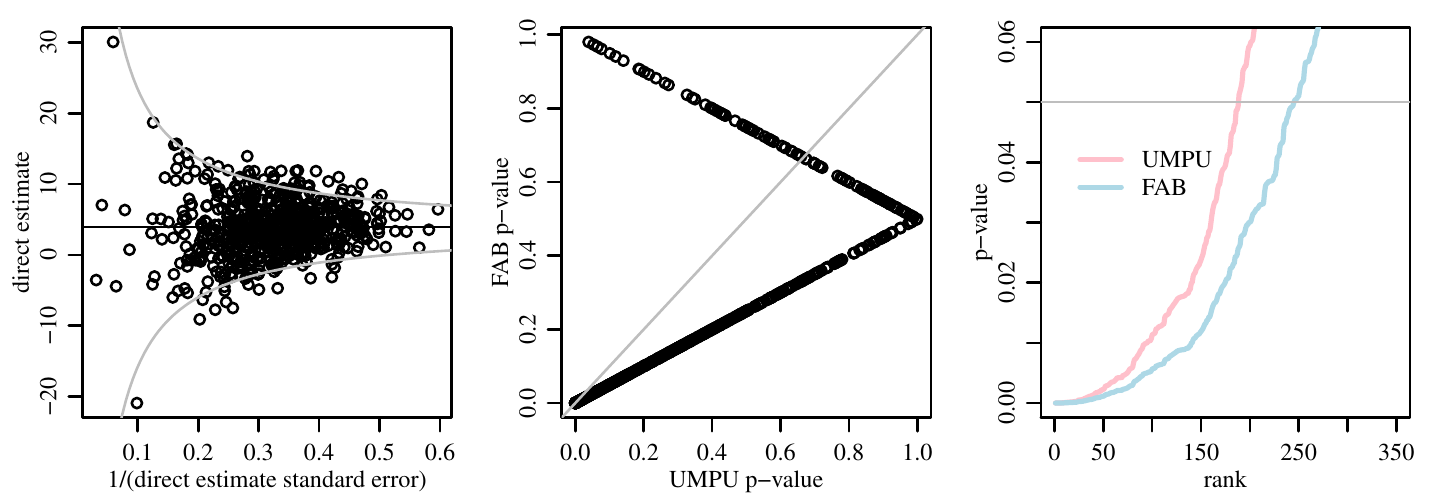} 

}

\end{knitrout}
\caption{Evaluation of school-specific associations between SES and standardized test performance.  The left panel plots OLS estimates versus their 
precision. The middle and right-most panels plot UMPU and FAB $p$-values. 
}
\label{fig:elsINT}
\end{center}
\end{figure}

\subsection{Asymptotically uniform null $p$-values}
In the preceding linear regression example, adaptive FAB 
$p$-values with exactly uniform null distributions are available 
because the covariance  of the observed data vector  is known 
up to a scalar. For each element of the data vector, this 
permits construction of an independent vector 
of indirect information. 
In cases where  the covariance 
is not sufficiently known to permit this decorrelation, 
an independent vector of indirect information 
may not be available. However, if a consistent estimator of the covariance 
is available then it is possible to construct approximate 
adaptive FAB $p$-values with null distributions that are asymptotically 
uniform.
Asymptotically uniform null $p$-values are also available in cases 
where the direct information is only asymptotically normally distributed. 
This may be the case if, for example, the direct estimates consist
 of regression coefficient estimates from a generalized linear model.

Consider the general case where 
the data available include a consistent and asymptotically 
normal estimator $\hat{\bs \theta}$ of $\bs\theta$,  so that 
$\sqrt{n}(\hat{\bs \theta}-\bs\theta) \stackrel{d}{\rightarrow} N_p(\bl 0,\bs\Sigma)$ as $n\rightarrow \infty$, 
and a consistent estimator $\hat{\bs\Sigma}$ of $\bs\Sigma$. 
The standard Wald  $p$-value for $\theta_j$
that is asymptotically uniform under $H_j:\theta_j=0$  can be written as
$1-| \Phi( \hat Z_j ) - \Phi(-\hat Z_j)|$, where 
$\hat Z_j = \sqrt{n} \hat \theta_j/\hat \sigma_j$.  
The proposed  approximate FAB $p$-value for $\hat\theta_j$ 
is  $1-| \Phi( \hat Z_j + \tilde b_j ) - \Phi(-\hat Z_j)|$, where
$\tilde b_j$ is obtained from 
indirect information about $\theta_j$ that is asymptotically independent of 
the direct information $\hat \theta_j$. 
This indirect information about $\theta_j$ can be constructed as follows:
Define the function $\bs\Gamma:\mathbb R^p \rightarrow \mathbb R^{p\times(p-1)}$
so that $\bs\Gamma(\bl s)$ is a $p\times (p-1)$ matrix whose columns form an orthonormal basis for the null space of $\bl s$. Specifically, define $\bs\Gamma(\bl s)$
so that the basis is obtained by
Gram-Schmidt orthogonalization of the
collection of vectors in $\mathbb R^p$  consisting of 
$\bl s$ and the standard basis vectors $\{ \bl e_k :  k \in \{1,\ldots, p\} \setminus \{j\} \}$. This construction implies that  
$\bs\Gamma$ is a continuous function of $\bl s$. 
Now let $\hat{\bl G}_j = \bs\Gamma( \hat{\bs\sigma}_j)$, 
where $\hat{\bs\sigma}_j$ is the $j$th column of $\hat{\bs\Sigma}$.  
By the continuous mapping theorem, $\hat{\bl G}_j$ is a consistent 
estimator of $\bl G_j = \bs\Gamma(\bs\sigma_j)$, where 
$\bs\sigma_j$ is the $j$th column of the true covariance matrix $\bs\Sigma$. 
The indirect information about $\theta_j$ can be derived from 
$\hat{\bl G}_j^\top \hat{\bs \theta}$. This random vector will not be exactly 
independent of the direct estimate $\hat \theta_j$   because 
$\hat{\bl G}_j$ is not exactly equal to $\bl G_j$ and $\hat{\bs\theta}$  might not be  exactly normally distributed, but $\hat{\bl G}_j^\top \hat{\bs \theta}$ should 
be approximately independent of $\hat \theta_j$ for large $n$, since 
$\hat{\bl G}_j$ is converging to $\bl G_j$, 
$\bl G_j^\top \hat {\bs \theta}$ is asymptotically uncorrelated with $\hat\theta_j$ and both are asymptotically normal. 
The following result shows that indeed 
$\hat{\bl G}_j^\top \hat {\bs \theta}$ is uncorrelated 
with and independent of $\hat \theta_j$ in an asymptotic sense.

\begin{thm}
Let $\hat{\bs\Sigma}\stackrel{p}{\rightarrow} \bs\Sigma$ and 
$\sqrt{n}(\hat{\bs\theta}-\bs\theta) \stackrel{d}{\rightarrow} 
         \bs E$, where $\bs E\sim N_p(\bl 0,\bs\Sigma)$. 
Then as $n\rightarrow \infty$, 
\begin{enumerate}
\item ${\rm Cor}[\hat{\bl G}_j^\top \hat{\bs \theta} , \hat \theta_j] \rightarrow 
\bl 0$; 
\item 
$\Pr( \{ \sqrt{n} \hat{\bl G}_j^\top (\hat{\bs\theta}-\bs\theta)\in A \}
  \ \cap \  
   \{    \sqrt{n}(\hat\theta_j - \theta_j) \in B \}) 
          \rightarrow  \Pr( \bl G^\top \bs E \in A ) \times  
     \Pr( E_j \in B)$,
\end{enumerate}
where 
$\bl G_j = \bs\Gamma(\bs\sigma_j)$, $\hat{\bl G}_j = \bs\Gamma(\hat {\bs\sigma}_j)$, and 
$A\subset \mathbb R^{n-1}$ and $B\subset \mathbb R$ are measurable sets.  
\label{thm:asyind}
\end{thm}
\begin{proof} 
 Let $\hat {\bl  H}_j$ 
be the $p\times p$ matrix obtained by binding the $j$th standard basis vector 
to $\hat{\bl G}_j$, and define $\bl H_j$ analogously. 
Since  $\bs\Gamma$ is continuous we have that $\hat{\bl G}_j \stackrel{p}{\rightarrow} 
 \bl G_j$
and also $\hat{\bl H}_j \stackrel{p}{\rightarrow} 
 \bl H_j$. 
Both results then follow by Slutsky's theorem, since 
$  \hat{\bl H}_j^\top \sqrt{n}(\hat{\bs\theta} - \bs\theta) \stackrel{d}{\rightarrow}
   {\bl H}_j^\top \bs E = ( \bl G_j^\top \bs E \ E_j)$, 
and $\bl G_j^\top \bs E$ is independent of $E_j$. 
\end{proof}

Now we construct a class of $p$-value functions that make use of the 
indirect information and are asymptotically uniform under each 
null hypothesis. Consider the $p$-value function  
$1-| \Phi (\hat Z_j + \hat {\sigma}_j f(\hat{\bl G}_j^\top \hat{\bs \theta}, 
   \hat{\bl G}_j)/\sqrt{n}) - \Phi(-\hat Z_j) |$ 
where 
  $f:\mathbb R^{p-1}  \times \mathbb R^{p\times (p-1)} \rightarrow\mathbb  R$
is a continuous function. For example, $f$ could be a function 
that, when input $(\bl G^\top_j \bs\theta,\bl G_j)$, returns 
$\Exp{ \theta_j | \bl G^\top_j \bs\theta }/\Var{\theta_j | \bl G_j^\top \bs\theta }$ based on a linking model $\bl G_j^\top \bs\theta\sim N_{p-1}(\bl G_j^\top\bs\theta, \bl G_j^\top \Psi \bl G_j)$. Since neither $\bl G_j$ nor $\bs\theta$ are known, 
we use plug-in estimates $\hat{\bl G}_j^\top \hat{\bs \theta}$ and $\hat{\bl G}_j$. 

An asymptotic analysis of this FAB $p$-value where 
$\bs\theta$ is fixed for all sample sizes is not particularly interesting. 
In this case, $f(\hat{\bl G}_j^\top \hat{\bs \theta}, \hat{\bl G}_j)$
converges to $f(\bl G_j^\top \bs\theta, \bl G_j)$ by the continuous mapping 
theorem, and so 
 the FAB adjustment 
  $\tilde b_j =  \hat \sigma_j f(\hat{\bl G}_j^\top \hat{\bs \theta}, 
   \hat{\bl G}_j)/\sqrt{n}$ will converge to zero and so too
the  difference between the FAB  and UMPU $p$-values.
More interesting is to consider a sequence of $\bs\theta$-values that decrease in magnitude with 
$n$, 
say $\bs\theta = \bs \theta_0/\sqrt{n}$ when the sample size is $n$.  
Now recall that 
$f(\bl G_j^\top \bs\theta ,\bl G_j)$ 
returns a plug-in estimate of  $\Exp{ \theta_j | \bl G^\top_j \bs\theta }/\Var{\theta_j | \bl G_j^\top \bs\theta }$ under 
the linking model $\bl G_j^\top \bs\theta\sim N_{p-1}(\bl G_j^\top\bs\theta, \bl G_j^\top \Psi \bl G_j)$. 
It is natural to use a function $f$ that is scale equivariant, which implies that 
$f$ satisfies
$f(c \bl G^\top \bs\theta ,\bl G)= 
     f(\bl G^\top \bs\theta ,\bl G)/c$ for  positive scalars $c$. 
In this case, we have $f( \hat{\bl G}_j^\top \hat{\bs \theta}, \hat{\bl G}_j)/\sqrt{n} =
    f( \hat{\bl G}_j^\top (\sqrt{n} \hat{\bs \theta}), \hat{\bl G}_j) $. 
 Now $\sqrt{n} \hat{\bs \theta}$ is approximately distributed as  
$N_p(\bs\theta_0 ,\bs \Sigma)$, 
while $\hat{\bl G}_j$ converges in probability to $\bl G_j$, so multiplication of 
 $\sqrt{n} \hat{\bs \theta}$ by $\hat {\bl G}_j$ gives an 
indirect information vector 
that is asymptotically independent of the direct estimate  $\hat\theta_j$, and 
so the resulting $p$-value will be asymptotically uniformly distributed if 
$\theta_{0j}=0$. 
This is  summarized
 as follows: 

\begin{figure}[ht]
\begin{knitrout}\footnotesize
\definecolor{shadecolor}{rgb}{0.969, 0.969, 0.969}\color{fgcolor}

{\centering \includegraphics[width=6in]{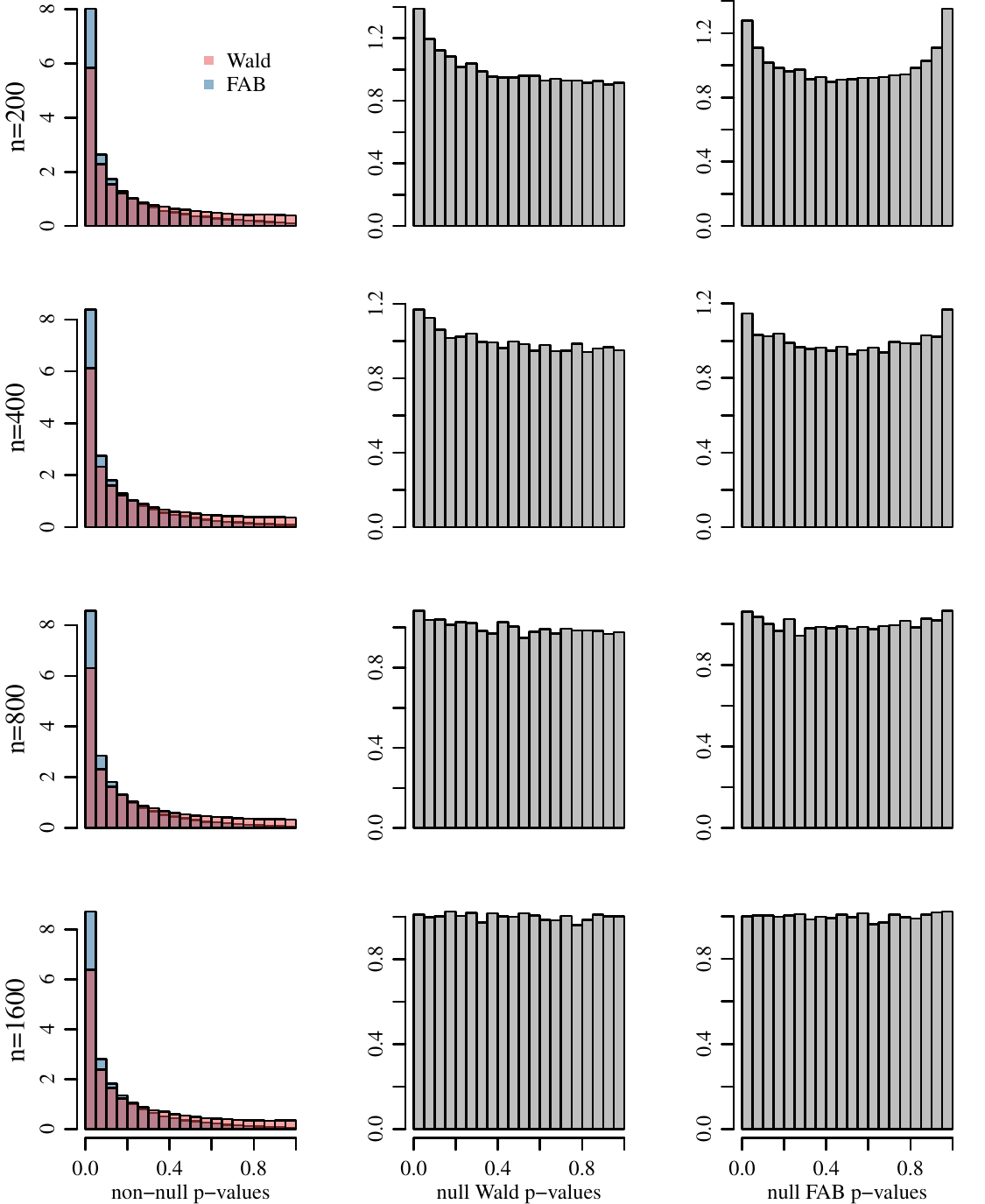} 

}

\end{knitrout}
\caption{Distributions of approximate $p$-values from the simulation study. 
 The first column displays distributions of Wald and FAB $p$-values for 
non-null parameters. The second and third columns display distributions of 
  $p$-values for null parameters. 
 }
\label{fig:exasym}
\end{figure}

\begin{thm}
Let $\hat{\bs\Sigma}\stackrel{p}{\rightarrow} \bs\Sigma$ and
$\hat{\bs \theta} \stackrel{d}{=} \bs Y/\sqrt{n} + o_p(1/\sqrt{n})$,
where $\bs Y\sim N_p(\bs\theta_0,\bs\Sigma)$.
Then as $n\rightarrow \infty$,
\begin{equation}
 1-| \Phi (\hat Z_j + \hat {\sigma}_j f(\hat{\bl G}_j^\top \hat{\bs \theta}, 
   \hat{\bl G}_j)/\sqrt{n}) - \Phi(-\hat Z_j) | \stackrel{d}{\rightarrow}
1- |\Phi(Z_j + \sigma_j f({\bl G}_j^\top \bs Y, \bl G_j))  
    - \Phi(-Z_j) |,   
\label{eqn:plimit}
\end{equation}
where $\hat Z_j = \sqrt{n} \hat\theta_j/\hat{\sigma}_j$,
$Z_j = Y_j/\sigma_j$, and
$f:\mathbb R^{p-1}  \times \mathbb R^{p\times (p-1)} \rightarrow\mathbb  R$
is a continuous function that satisfies
$f( c \bl a ,\bl B) = f(\bl a , \bl B)/c$ for $c>0$.
The right side of (\ref{eqn:plimit}) is
uniformly distributed if $\theta_{0j}=0$.
\end{thm}
\begin{proof}
The $p$-value function $p(z,b) = 1- | \Phi(z+b) - \Phi(-z) |$ is
continuous in $z$ and $b$, and so by the continuous mapping theorem
it suffices to show that
$\hat Z_j \stackrel{d}{\rightarrow} Z_j$
and $\hat \sigma_j f(\hat{\bl G}_j^\top \hat{\bs \theta}, 
\hat{\bl G}_j)/\sqrt{n}  \stackrel{d}{\rightarrow} 
 \sigma_j f({\bl G}_j^\top \bs Y, \bl G_j)$.
The former convergence follows from the assumptions and Slutsky's theorem.
For the latter, note that by the equivariance of $f$, we have that
$f(\hat{\bl G}_j^\top \hat{\bs \theta}, 
\hat{\bl G}_j)/\sqrt{n}$ is equal to
$f(\hat{\bl G}_j^\top (\sqrt{n}\hat{\bs \theta}), 
\hat{\bl G}_j)$ which converges to
$f( {\bl G}_j^\top \bs Y,\bl G_j)$ by the continuous mapping theorem and
that $\hat{\bl G}_j^\top (\sqrt{n}\hat{\bs \theta})\stackrel{d}{\rightarrow} 
  {\bl G}_j^\top \bs Y$ (which follows from the consistency of $\hat{\bl G}_j$
for $\bl G_j$ and Slutsky's theorem).
\end{proof}

We evaluate  this result empirically in a simulation study with 
$p=30$ and four sample sizes $n\in \{ 200,400,800,1600\}$. 
For each sample size $n$, 
5000 binary vectors $\bs Y\in\mathbb \{0,1\}^n$ were constructed 
with elements simulated as 
$Y_i \sim \text{binary}( e^{\bs\theta^\top\bl x_i}/(1+ e^{\bs\theta^\top\bl x_i} )$, independently for $i=1,\ldots, n$. 
The $p$-dimensional vector of 
regression coefficients $\bs\theta$ had  15 elements equal to $3/\sqrt{n}$
and 15 elements equal to zero, and the 
elements of each $\bl x_i\in \mathbb R^p$ were simulated independently from a  standard normal distribution. 
For each simulated $\bs Y$ vector, the MLE $\hat{\bs\theta}$ 
was obtained, along 
with an estimate $\hat{\bs \Sigma}/n$ of its variance based on the observed information matrix. 
For each estimated coefficient $\hat\theta_j$, 
a $Z$-statistic $\hat Z_j = \sqrt{n}\hat\theta_j/{\hat\sigma_j}$ was computed
and used to construct 
both a Wald and FAB  $p$-value. 
The FAB $p$-value was constructed using a linking model where $\bs\theta$ is an i.i.d.\ sample from a  mixture of a $N(\mu,\tau^2)$ 
distribution and a point-mass at zero. For each $j$, 
the parameters $\mu$, $\tau^2$ and the mixture proportions were estimated 
using marginal maximum likelihood and the 
indirect information $\hat{\bl G}_j^\top \hat{\bs\theta}$. 
The FAB $p$-value under this linking model 
is 
$1-|\Phi( \hat Z_j + 2 \hat\sigma_j\tilde\mu_j /[\sqrt{n}\tilde\tau^2_j])-
   \Phi(-\hat Z_j)|$, 
where $\tilde\mu_j$ and $\tilde\tau^2_j$ are estimated from the linking model
using $\hat{\bl G}_j^\top \hat{\bs\theta}$. 

Some results are shown in Figure \ref{fig:exasym}. The first column of the 
figure shows histograms of the Wald and FAB $p$-values corresponding to 
non-null coefficient values. As can be seen, the FAB $p$-values tend to 
be smaller, consistently across the four sample sizes. 
For example, the fraction of non-null FAB $p$-values below 0.05 
was 
0.4, 0.42, 0.43 and 0.44
for $n\in \{200,400,800,1600\}$ as  compared to 
0.29, 0.31, 0.32 and 0.32 
for the Wald $p$-values. 
The second and third columns of the figure plot 
histograms of the $p$-values corresponding to zero coefficient values. 
Neither the Wald 
nor the FAB null $p$-values are exactly uniformly distributed, however 
the uniformity of these distributions increases with $n$. Interestingly, 
the shape of the null $p$-value distributions is slightly different for 
the smaller sample sizes, although near zero the shapes are similar. 
For example, the fraction of null FAB $p$-values below 0.05
was
0.06, 0.06, 0.05 and 0.05
for $n\in \{200,400,800,1600\}$ as  compared to
0.07, 0.06, 0.05 and 0.05
for the Wald $p$-values.

\section{Discussion}
The FAB $p$-value 
is derived from a
test statistic that has optimal average power  with respect to a  
probability distribution over the possible values of the mean  
of a normal distribution. 
Ideally, this probability distribution over the normal mean
should represent
prior or indirect information. 
However, regardless of the distribution used to construct it, 
the FAB $p$-value 
is uniformly distributed under the null hypotheses. 

In multiparameter settings, 
indirect information about one parameter may be derived from 
data on the other parameters, and then used to construct an adaptive 
FAB $p$-value. One way to do this is with a 
linking model that 
relates the parameters to each other. A Gaussian linking model 
is convenient, in terms of both the estimation of the parameters 
in the linking model, and the simplicity of the FAB $p$-value 
when the indirect information is in the form of 
 a normal distribution. However, other linking models could certainly 
be used, which would yield different optimal test statistics and 
different $p$-value functions. 

Related to this, an  interesting question is the extent to which 
a FAB $p$-value 
based on a Gaussian linking model has good performance
on average with respect to a non-Gaussian linking model. 
It seems plausible  
that in many scenarios, even if the true values of the parameters are better 
represented by a non-Gaussian linking model, adaptive $p$-values
constructed using a Gaussian linking model will at least perform 
better than the UMPU $p$-values, on average across the parameters 
(this was the case for the example in Section 3.3, which used a misspecified linking model). 
To see why, consider the 
simplest case where the direct data for parameter $\theta_j$ is $Z_j\sim N(\theta_j,1)$. By 
Theorem 
\ref{prop:pbetter}, the probability that 
the FAB $p$-value is smaller than the UMPU $p$-value is 
at least  $\Phi( \text{sign}( \tilde b_j ) \times \theta_j)$, 
where $\tilde b_j$ is two times the 
ratio of the expected value  
of $\theta_j$ to its variance, as estimated from the linking model.
So roughly speaking, as long as the linking model is 
good enough so that $\Phi( \text{sign}( \tilde b_j ) \times \theta_j)$ 
is greater than 1/2 on average across $\theta_j$'s, 
then a majority of 
the adaptive FAB $p$-values should be lower than the corresponding 
UMPU $p$-values. 
The most challenging case for the FAB procedure is perhaps 
when the $\theta_j$'s 
are unrelated and centered around zero. However, if the linking model 
contains the mean-zero normal distributions with covariance proportional 
to the identity matrix
(as it did in each example presented in this article), 
then the lack of structure 
among the $\theta_j$'s should be reflected in the estimates of the parameters
in the linking model. In this case, the values of 
the $\tilde b_j$'s should be close to zero and the 
 FAB $p$-values will be approximately the same as the UMPU $p$-values. 

This article has focused on  constructing adaptive 
FAB $p$-values in a variety of multiparameter settings, 
rather than recommending how they might be used. In some settings, such as
the analysis of the ELS data presented in Sections 3.2 and 4.2, it is 
the school-specific $p$-values themselves that may be of interest, 
as the faculty of a given school are presumably primarily concerned 
with, for example, the relationship between SES and test scores in 
their own school rather than what this relationship is on average across 
schools. In other settings it may be of more interest to use the 
$p$-values as inputs into other procedures that maintain global 
error rates. In such cases, it must be remembered that the 
adaptive FAB $p$-values are dependent by construction, and so procedures 
that rely on independent $p$-values must be used with caution. 
However, for most scenarios we expect that this dependence among FAB 
$p$-values will be positive dependence of 
some sort, 
in which case results of 
\citet{clarke_hall_2009} suggest that target error rates 
will still be approximately maintained for some popular 
procedures designed for independent $p$-values. 
Alternatively, FAB $p$-values could be used as inputs into procedures 
that accommodate, and even make use of, the dependence among $p$-values
\citep{efron_2007,romano_shaikh_wolf_2008,sun_cai_2009,fan_han_gu_2012}. 

\medskip

Replication materials, additional numerical examples and software to compute FAB $p$-values for parameters in linear and generalized linear models are available at 
 \url{https://pdhoff.github.io/FABInference}

\bibliography{fabPval}

\end{document}